\documentclass[reqno]{amsart}
\usepackage[foot]{amsaddr}
\usepackage{amsmath,amssymb,amsthm}
\usepackage{graphics}
\usepackage{tikz}

\newtheorem{theorem}{Theorem}
\newtheorem{lemma}[theorem]{Lemma}
\newtheorem{corollary}[theorem]{Corollary}
\newtheorem{claim}[theorem]{Claim}

\theoremstyle{definition}
\newtheorem{definition}{Definition}
\newtheorem*{method}{Maximum Entropy Method}

\renewenvironment{cases}{\left\{\mkern-6mu\begin{array}{l@{~~~}l}}{\end{array}\right.}

\let\norm\|
\newcommand\tsp{\mkern 1.5mu}
\def\|{\mkern1.5mu{|}\mkern1.5mu}
\def\restr{\mathord\upharpoonright}
\renewcommand\ge{\geqslant}\renewcommand\le{\leqslant}

\newcommand\C{\mathcal C}
\newcommand\Co{\mathsf C} 
\newcommand\D{\mathcal D}
\newcommand\I{\mathcal I} 
\newcommand\J{\mathcal J} 
\newcommand\Z{\mathcal Z} 
\newcommand\N{\mathbb N}  
\newcommand\R{\mathbb R}  
\newcommand\T{\top}       
\newcommand\cups{\mathrel{\cup\mkern -1mu^*\mkern-3mu}}
\def\<{\langle}\def\>{\rangle}

\renewcommand\phi{\varphi}
\DeclareMathAlphabet{\mathbfit}{OML}{cmm}{b}{it}
\def\r{\mathbfit r}
\def\b{\mathbfit b}
\def\Hen{\mathbf H}
\newcommand\eqdef{\stackrel{\mathrm{def}}{=}}
\newcommand\dn{\mathord{\downarrow}}
\newcommand\tsum{{\textstyle\sum}}
\renewcommand\cups{\mathop{{\cup}\mkern -2mu ^*}}
\DeclareMathOperator\len{\mathrm{len}}

\def\m#1{\tsp{#1}\tsp}

\def\sm{\tsp{-}\tsp} 
\def\titem#1{\noindent\hbox to 2\parindent{\hss\upshape
#1}\space\ignorespaces\advance\leftmargin2\parindent}

\def\G#1{\Gamma_{\!#1}}
\def\Ga#1{\Gamma^*_{\!#1}}
\def\barGamma{\mkern1mu\overline{\hbox{$\mkern-1mu\Gamma\mkern4mu$}}\mkern-4mu}
\def\clGa#1{\barGamma^*_{\!#1}}

\def\abcd{[abcd\tsp]}
\def\acbd{[acbd\tsp]}
\def\cdab{[cdab\tsp]}
\def\YXZ{Y\mkern-2mu XZ}
\def\YX{Y\mkern-2mu X}

\makeatletter
\def\@setemails{%
    \mbox{{\itshape E-mail address}:\space}{\ttfamily\emails}.%
}
\makeatother
\usepackage{newfloat}
\usepackage[noEnd=false,indLines=true,beginComment=~,commentColor=black!90,
   beginLComment=--~,endLComment=]{algpseudocodex}
\tikzset{algpxIndentLine/.style={draw,line width=0.5pt,color=gray!60}}
\makeatletter
\algnewcommand\algorithmicforeach{\textbf{for each}}
\algdef{S}[FOR]{ForEach}[1]{%
        \algpx@startIndent\algpx@startCodeCommand\algorithmicforeach\ #1\ \algorithmicdo%
}
\pretocmd{\ForEach}{\algpx@endCodeCommand}{}{}
\algdef{SE}[LOOP]{Loop}{EndLoop}{%
       \algpx@startIndent\algpx@startCodeCommand\algorithmicloop%
}{%
        \algpx@endIndent\algpx@startEndBlockCommand\algorithmicend%
}
\algdef{SE}[IF]{If}{EndIf}[1]{%
        \algpx@startIndent\algpx@startCodeCommand\algorithmicif\ #1\ \algorithmicthen%
}{%
   \algpx@endIndent\algpx@startEndBlockCommand{\algorithmicend\algorithmicif}%
}
\algdef{SE}[FOR]{For}{EndFor}[1]{%
        \algpx@startIndent\algpx@startCodeCommand\algorithmicfor\ #1\ \algorithmicdo%
}{%
   \algpx@endIndent\algpx@startEndBlockCommand{\algorithmicend\algorithmicfor}%
}
\algdef{SE}[WHILE]{While}{EndWhile}[1]{%
    \algpx@startIndent\algpx@startCodeCommand%
    \algorithmicwhile\ #1\ \algorithmicdo%
}{%
     \algpx@endIndent\algpx@startEndBlockCommand\algorithmicend}%

\makeatother
\algrenewcommand\alglinenumber[1]{\scriptsize #1~}
\DeclareFloatingEnvironment[name=Algorithm]{code}
\newenvironment{pseudocode}[1]
{\begin{code}[htb]%
  \hrule\vskip5pt
  \caption{\fontsize{9}{10}\selectfont #1}%
  \vskip 3pt
  \hrule\hbox{}\begin{algorithmic}[1]\ignorespaces}
{\end{algorithmic}
  \vskip 3pt \hrule
  \end{code}}

\usepackage[bookmarks=true,hyperindex=true]{hyperref}
\hypersetup{
  pdftitle={Information Inequalities for Five Random Variables},
  pdfauthor={E.P. Csirmaz and L. Csirmaz},
  pdfkeywords={entropy; information inequalities; polymatroid; polyhedral geometry},
  pdfsubject={Mathematics; Information Theory},
  bookmarksopen = false,
  pageanchor = false,
  plainpages = false,
}
\newcommand\bref[1]{\hyperlink{#1}{#1}}
\newcommand\blabel[1]{\hypertarget{#1}{}}

\title{Information Inequalities for Five Random Variables}
\author{E.~P.~Csirmaz${\strut}^*$}
\author{L.~Csirmaz${\strut}^{*\dagger}$}
\address{\normalfont ${\strut}^{*}$ R\'enyi Institute, Budapest}
\address{\normalfont ${\strut}^\dagger$ UTIA, Prague}
\email{csirmaz@renyi.hu}
\date{}
\begin{document}
\def\r{\mathbfit r}
\def\b{\mathbfit b}
\advance\abovedisplayskip.3\baselineskip
\advance\belowdisplayskip.3\baselineskip
\makeatletter
\let\atopwithdelims\@@atopwithdelims
\def\section{\@startsection{section}{1}%
  \z@{1.0\linespacing\@plus\linespacing}{.7\linespacing}%
  {\normalfont\scshape\centering}}
\makeatother
\begin{abstract}

The entropic region is formed by the collection of the Shannon entropies of
all subvectors of finitely many jointly distributed discrete random
variables. For four or more variables, the structure of the entropic region
is mostly unknown. We utilize a variant of the Maximum Entropy Method to
obtain five-variable non-Shannon entropy inequalities, which delimit the
five-variable entropy region. This method adds copies of some of the random
variables in generations. A significant reduction in computational
complexity, achieved through theoretical considerations and by harnessing
the inherent symmetries, allowed us to calculate all five-variable
non-Shannon inequalities provided by the first nine generations. Based on
the results, we define two infinite collections of such inequalities and
prove them to be entropy inequalities. We investigate downward-closed
subsets of non-negative lattice points that parameterize these collections,
and based on this, we develop an algorithm to enumerate all extremal
inequalities. The discovered set of entropy inequalities is conjectured to
characterize the applied method completely.

\smallskip


\smallskip
\noindent{\bf AMS Classification Numbers:}
05B35; 
26A12; 
52B12; 
90C29; 
94A17; 
52B40; 
90C27 

\end{abstract}

\maketitle

\section{Introduction}

Many important mathematical problems can be reduced to the following
question: does a collection of finite random variables exist such that
the entropies of the variable subsets satisfy certain linear constraints?
Examples include, but are not limited to, channel coding \cite{Csiszar.Korner}
and network coding in particular \cite{yeung-network}, estimating the
efficiency of secret sharing schemes \cite{Beimel,Beimel.Orlov,G.Rom},
questions about matroid representations \cite{BBFP}, guessing games
\cite{riis,rombach}, extracting information from common strings in cryptography
\cite{groth}, additive combinatorics \cite{madiman-etal}, and finding
conditional independence inference rules \cite{studeny}.

The \emph{entropy function} of finitely many discrete random variables
$\langle \xi_i:i\in N\rangle$ indexed by the fixed finite set $N$ maps the
non-empty subsets $I\subseteq N$ to the Shannon entropy $\Hen(\xi_I)$ of the
variable set $\xi_I=\langle x_i:i\in I\rangle$, see \cite{yeung-book}. The
\emph{entropy region}, denoted by $\Ga N$, is the range of the entropy
function; it is a part of the $2^{|N|}\m-1$-dimensional Euclidean space where
the coordinates are labeled by non-empty subsets of $N$. Entropies are
non-negative real numbers, and thus the entropy region lies in the
non-negative orthant of this Euclidean space. It is delimited by a
collection of homogeneous linear
inequalities corresponding to the non-negativity of basic Shannon
information measures \cite{yeung-book}. Points satisfying all these
inequalities form the \emph{Shannon-bound}; the Shannon-bound is denoted by $\G N$.

N.~Pippenger argued in \cite{pippenger} that linear inequalities bounding
the entropic region $\Ga N$ encode the fundamental laws of Information
Theory and determine the limits of information transmission and data
compression. The long-standing problem of whether a linear information
inequality can properly cut into the Shannon bound was settled in 1998 by
Zhang and Yeung \cite{ZhY.ineq} by exhibiting the first example of such a
non-Shannon information inequality. Their discovery initiated intensive
research. The phrase \emph{Copy Lemma} was coined by Dougherty et
al.~\cite{DFZ11} to describe the general method distilled from the original
Zhang--Yeung construction. The Copy Lemma has been applied successfully to
generate several hundred sporadic and a couple of infinite families of
non-Shannon entropy inequalities for $\Ga4$, see
\cite{csirmaz.book,DFZ11,M.infinf}. A different method, utilizing an
information-theoretic lemma attributed to Ahlswede and K\"orner \cite{AK},
was proposed in \cite{MMRV}; later it was shown to be equivalent to a
special case of the Copy Lemma \cite{Kaced}.

Our method to obtain five-variable non-Shannon entropy inequalities is based
on a more general paradigm of which the Copy Lemma is a special case
\cite{exploring}. Derived from the principle of maximum entropy
\cite{maxentp}, it is called MEM, short for Maximum Entropy Method. For more
details, see Section \ref{sec:method}.

Previous works on generating and applying non-Shannon entropy inequalities,
such as \cite{Beimel.Orlov,csirmaz.book,inner,DFZ11,entreg,studeny}, focused
on the four-variable case, and only a few sporadic five-variable non-Shannon
inequalities have been discovered, such as the MMRV inequality from
\cite{MMRV} or an infinite set of inequalities from \cite{M.infinf}. This is
the first work that provides a method that systematically generates an
infinite collection of non-Shannon bounds on the five-variable entropy
region $\Ga5$. Compared to the four-variable case, there are significant
challenges, both theoretical and computational. The four-variable entropy
region $\Ga4$ sits in the $15$-dimensional Euclidean space, while the
five-variable region $\Ga5$ is $31$-dimensional. The structure of the
Shannon bound $\G4$ is well-understood: it has $41$ extremal directions, and
only $6$ of them have no entropic points. The entropy region $\Ga4$ has an
inner polyhedral cone where it fills its Shannon bound, and has six
isomorphic ``protrusions'' towards the six exceptional extremal directions,
each protrusion surrounded by $15$ hyperplanes of which $14$ come from the
Shannon bound \cite{Ma.Stud}. Only the protrusions contribute to new entropy
inequalities, and their dimension can be reduced to $10$. Computational
results about $\Ga4$ can be obtained by computing vertices and facets of
numerous implicitly defined $10$-dimensional polyhedra \cite{inner}. In
contrast, the Shannon bound $\G5$ of the five-variable entropy region has
117,983 extremal directions \cite{studeny-kocka}, and for a few of them it
is not even known whether they contain an entropic point or not. No
structural reduction similar to the four-variable case is available, and it
is not known whether such a reduction exists or not. Computations about
$\Ga5$ can still be reduced to $25$-dimensional polyhedral enumeration
problems (although with significantly larger number of constraints than in
the $4$-variable case). The complexity of enumeration problems typically
doubles when the dimension increases by one, making such high-dimensional
enumeration problems practically intractable.

We overcome this computational difficulty by applying a particular variant
of the Maximum Entropy Method. This variant, working in generations, first
reduces the problem dimension from $31$ to $19$, and then, at each
generation, adds extra copies of some of the random variables, increasing
the problem dimension again. Theoretical considerations and harnessing the
inherent symmetry allowed us to complete the associated polyhedral
computations up to nine generations. The output was the complete list of
five-variable non-Shannon inequalities provided by the first nine
generations. Based on the experimental results, we define an infinite
collection of five-variable inequalities that we \emph{prove} to be provided
by this MEM variant---in particular, they are valid non-Shannon entropy
inequalities---and \emph{conjecture} this collection to be complete; that
is, no additional inequalities are yielded by this MEM variant. The
collection of the inequalities is parametrized by finite, downward closed
subsets of the non-negative lattice points of the plane. Some of the
inequalities in our collection are consequences of the others; those that
are not, are called \emph{extremal}. We developed an incremental algorithm
that enumerates, from generation to generation, the parameters yielding the
extremal inequalities, in complete agreement with the computational results.
The algorithm allowed us to significantly exceed the capabilities of
polyhedral computation. While numerical instability prevented the completion
of the polyhedral computation for the tenth generation, all extremal entropy
inequalities were enumerated up to generation 60. Finally, we have looked at
the large-scale behavior of the extremal inequalities, and pictured how
these inequalities delimit a three-dimensional cross-section of $\Ga5$.

The new five-variable non-Shannon inequalities can be applied to real-world
problems. The~most immediate application is in network coding. The~new
inequalities tighten the boundaries; they provide stricter and more accurate
bounds on network capacity. In~a network protocol they can assist in
proving whether a targeted data rate is achievable or not~\cite{yeung-network}.

Cloud storage services (like Google Drive or AWS S3) distribute data
fragments across many nodes~\cite{cloud}. In~case of failure, the~node has
to download the missing data from other nodes. The~new five-variable inequalities
can be used to determine the theoretical limits of storage efficiency for
systems with more complex failure models or larger~clusters.

In the realm of secret sharing, entropy inequalities provide lower bounds on
the size of secrets~\cite{Beimel,Beimel.Orlov,G.Rom}. To~explore another
facet of these problems, the~new inequalities can prove that certain
efficient schemes are impossible to realize. In~complex datasets, it is
important to distinguish between correlation and actual
causation~\cite{casual}. When an AI model analyzes data to build a causal
graph, it can use entropy inequalities to rule out models that are
information-theoretically impossible, narrowing the search space and
improving~accuracy.

In this paper lemmas, claims and theorems are arranged so that each is used
in the same section only, typically right after they are stated and proved.
Section \ref{sec:prepare} proves structural properties of the entropy region
that are used to reduce the computational complexity of the polyhedral
algorithms. The main theoretical results are stated and proved in Section
\ref{sec:all}. In Theorem \ref{thm:mainI} we prove a large collection of
entropy inequalities parametrized by the downward-closed subsets of the
non-negative lattice points. Lemmas estimating different entropy expressions
are used in this section only. Unfortunately, many inequalities provided by
Theorem \ref{thm:mainI} are consequences of the others. Claims and lemmas in
Section \ref{sec:reducedset} provide the theoretical foundation for our
algorithm that selects and enumerates the \emph{extremal} inequalities
among~them.

\smallskip

The remaining part of the paper is organized as follows. Notations are
recalled in Section \ref{sec:prelim}. Section \ref{sec:method} describes the
special variant of the Maximum Entropy Method we apply to $\Ga5$. Section
\ref{sec:prepare} discusses possible simplifications, including how symmetry
can be utilized and how the MEM parameters were chosen. Section
\ref{sec:compresults} describes the chosen coordinate systems, polyhedral
computations, and their results. Section \ref{sec:ineq} presents the
five-variable inequalities we obtained, paving the way for the definition of
two infinite families of such inequalities in Section \ref{sec:all}.
Additional theoretical results, including the proof that inequalities in
these families are indeed generated by the MEM method, are presented in
Section \ref{sec:all}. Section \ref{sec:reducedset} discusses methods that
can recognize extremal inequalities and the incremental algorithm that
enumerates the extremal inequalities for each MEM generation, describes the
large-scale behavior of the new inequalities, and investigates the delimited
part of the five-variable entropy region. In Section \ref{sec:fminf},  we take
a look at the five-variable non-Shannon entropy inequalities constructed by
F.~Mat\'u\v s in \cite{M.infinf} and compare them to a special subclass of
the inequalities obtained in Section \ref{sec:all}. Finally, Section
\ref{sec:final} summarizes our work, lists open questions, and provides
directions for further research.
 

\section{Preliminaries}\label{sec:prelim}

In this paper all sets are finite. Capital letters, such as $A$, $J$, $N$,
etc., denote (finite) sets; elements of these sets are denoted by lower case
letters. The union sign and the curly brackets around singletons are
frequently omitted, thus, $Nij$ denotes the set $N\cup\{i,j\}$. The
difference of two sets is written as $A\m-B$, or $A\m-b$ if the second set
is a singleton. The star in the union $A\cups B$ emphasizes that $A$ and $B$
are disjoint sets. A partition of $N$ is a collection of non-empty disjoint
subsets of $N$ whose union equals $N$.

A \emph{discrete random variable $\xi$} takes its values from a finite set
$\mathcal X$, called \emph{alphabet}. The probability that $\xi$ takes
$x\in\mathcal X$ is denoted by $\Pr(\xi=x)$, or simply by $\Pr(x)$ when the
random variable $\xi$ is clear from the context. Suppose $\xi$ is defined on
the direct product $\mathcal X=\prod_{i\in N} \mathcal X_i$ for some finite
set $N$, called the \emph{base set}. For a non-empty $A\subseteq N$ the
\emph{marginal $\xi_A$} is defined on the product alphabet $\mathcal X_A =
\prod_{i\in A}\mathcal X_i$ so that the probability of $y\in \mathcal X_A$
is the sum of the probabilities of those $x\in\mathcal X$ whose projection
to $\mathcal X_A$ equals $y$:
$$
   \Pr(y) = \sum \{ \Pr(x): x\restr A = y \}.
$$
To emphasize that $\xi$ is defined on a product space, we write
$\xi=(\xi_i:i\in N)$, and say that the random variables $\xi_i$ are
\emph{distributed jointly}. The Shannon entropy of the distribution
$\xi$ is defined as
$$
   \Hen(\xi) = \sum_{x\in\mathcal X} -\Pr(x)\log \Pr(x)
$$
with the convention that $0\log 0 = 0$. If $\xi=(\xi_i:i\in N)$ is a joint
distribution, then we write $\Hen_\xi(A)$ for $\Hen(\xi_A)$. The index $\xi$
is also dropped when it is clear from the context. By convention,
$\Hen(\emptyset)=0$. The entropies $\Hen_\xi(A)$ are arranged into a vector
indexed by the non-empty subsets $A$ of $N$. This vector is the
\emph{entropy profile} of the distribution $\xi$. The collection of these
$(2^{|N|}\m-1)$-dimensional vectors forms the \emph{entropy region}, denoted
by $\Ga N$. Elements of $\Ga N$ are considered interchangeably as vectors,
as points in this Euclidean space, and as functions assigning non-negative
real numbers to non-empty subsets of the base set $N$. For a gentle
introduction to these notions of Information Theory, please consult
\cite{yeung-book}.

Notions of conditional entropy, mutual information, and conditional mutual
information from Information Theory are formally extended to the functional
form of these vectors. If $f$ is any function on subsets of $N$, then for
subsets $A,B,C,D$ of $N$ the following forms will be used as abbreviations:
\begin{align*}
   f(A\|B) &\eqdef f(AB)-f(B), \\
   f(A,B)  &\eqdef f(A)+f(B)-f(AB), \\
   f(A,B\|C) &\eqdef f(AC)+f(BC)-f(ABC)-f(C), \mbox{ and}\\
   f[A,B,C,D] &\eqdef -f(A,B)+f(A,B\|C)+f(A,B\|D)+f(C,D).
\end{align*}
The first three expressions are called \emph{conditional entropy},
\emph{mutual information}, and \emph{conditional mutual information},
respectively. The last line defines the \emph{Ingleton expression}. An
entropy function is not defined on the empty set, nevertheless,
$f(\emptyset)=0$ will be assumed whenever convenient. In particular,
$f(A,B\|\emptyset)$ and $f(A,B)$ are the same expressions. Frequently, when
clear from the context, the function $f$ is omitted before the parenthesized
expression. Additionally, if applied to singletons, the Ingleton expression
is written without commas. An example is the inequality
\begin{equation}\label{eq:sampleineq}
  \abcd+(a,b\|z)+(b,z\|a)+(a,z\|b)+3(z\|ab) \ge 0.
\end{equation}

The \emph{Shannon inequalities} state the non-negativity of the conditional
entropy, mutual information, and conditional mutual information for all
subsets $A$, $B$, $C$ of the base set $N$. They are consequences of the
unique minimal set of such inequalities, called \emph{basic Shannon
inequalities}, see \cite{yeung-book}, listed in (\bref{B1}) and (\bref{B2})
below:
\begin{itemize}\setlength\itemsep{3pt plus 1pt minus 2pt}
\item[(B1)]\blabel{B1}$f(i\|N\sm i)\ge 0$ for all $i\in N$;
\item[(B2)]\blabel{B2}$f(a,b\|K)\ge 0$ for all $K\subseteq N$ and
different $a,b\in N\sm K$, including $K=\emptyset$.
\end{itemize}
The collection of all $(2^{|N|}\m-1)$-dimensional vectors (or points, or
functions) that satisfy the Shannon inequalities is denoted by $\G N$. It is
a natural outer bound for the entropy region $\Ga N$. $\G N$ is a pointed
polyhedral cone \cite{ziegler}; its facets are the hyperplanes specified by
the basic Shannon inequalities. \emph{Polymatroids} are elements of $\G N$
written in functional form. A polymatroid is usually written as $(f,N)$, or
just $f$, and we say that \emph{$f$ is on $N$}. The polymatroid $f$ is
\emph{entropic} if it is in $\Ga N$, and \emph{almost entropic}, or
\emph{aent} for short, if it is in the closure (in the usual Euclidean
topology) of $\Ga N$. Linear inequalities valid for all polymatroids are
consequences of the basic Shannon inequalities; an example is the inequality
(\ref{eq:sampleineq}). A \emph{non-Shannon inequality} is a homogeneous
linear inequality that is valid for points of the entropic region but not
for all points of the Shannon bound. Equivalently, the non-negative side of
the hyperplane corresponding to such an inequality contains the complete
entropy region, while it cuts properly into $\G N$.

The closure of the entropic region is a pointed convex full-dimensional
cone \cite{yeung-book}, and only its boundary points can be non-entropic
\cite{M.fmadhe}.
 
\smallskip
The polymatroid $(f,N)$ on the base set $N$ is \emph{linearly representable
over the field $\mathbb F$}, or \emph{ $\mathbb F$-representable} in short,
if there is a finite-dimensional vector space $V$ over $\mathbb F$, and
linear subspaces $V_i\subseteq V$ for $i\in N$, such that for all
$I\subseteq N$, $f(I)$ is the dimension of the linear subspace spanned by
$\bigcup_{i\in I} V_i$. Clearly, if both $(f,N)$ and $(g,N)$ are $\mathbb
F$-representable over the same field, then so is their sum $f+g$. The
polymatroid $f$ is \emph{$\mathbb F$-linear} if it is in the closure of the
multiplies of $\mathbb F$-representable polymatroids. By the previous
remark, $\mathbb F$-linear polymatroids form a closed cone. Finally, $f$ is
\emph{linear}, if it is $\mathbb F$-linear for some field $\mathbb F$.

Following a compactness argument, if $f$ is $\mathbb F$-representable, then
it is representable over some finite field as well, see \cite{jbell},
meaning that the vector space $V$ is also finite. Taking the uniform
distribution on $V$ provides the entropic polymatroid $(\log |V|)f$. Thus,
linear polymatroids are also almost entropic.

Linear polymatroids on the base set $N$ with $|N|\le 5$ are $\mathbb
F$-linear for every field $\mathbb F$, see \cite{Ma.Stud} and \cite{DFZ10};
this statement is not true in general. For
$|N|\le3$ every polymatroid is linear. For $N=\{abcd\}$ a polymatroid $f$ on
$N$ is linear if and only if it satisfies the following six instances of the
Ingleton inequality:
$$\begin{array}{ccc}
   f\abcd\ge 0, & f\acbd\ge 0, & f[adbc]\ge 0, \\[2pt]
   f[bcad\tsp]\ge 0, & f[bdac]\ge 0,     & f\cdab\ge 0,
\end{array}$$
see \cite{Ma.Stud}.
Since the Ingleton expression is symmetric in the first two and in the
last two arguments, these expressions cover all $24$ permutations of
$N$.

\smallskip
Finally, we recall notions of independence. Let $(f,N)$ be a polymatroid,
and $X$, $Y_1$, \dots, $Y_k$ be disjoint subsets of $N$. $Y_1$ and $Y_2$ are
\emph{independent in $f$} if $f(Y_1,Y_2)=0$. The collection $Y_1,\dots,Y_k$
is \emph{completely independent} in $f$ if for any two disjoint subsets $I$
and $J$ of the indices $\{1,2,\dots,k\}$, $Y_I=\bigcup_{i\in I}Y_i$ and
$Y_J$ are independent, or, equivalently, if
$$ 
   f(Y_1\cdots Y_k) = f(Y_1)+\cdots +f(Y_k).
$$
In this case we also have $f(Y_I)=\sum_{i\in I} f(Y_i)$ for every subset $I$
of the indices. The disjoint subsets $Y_1$ and $Y_2$ are \emph{conditionally
independent over $X$} if $f(Y_1,Y_2 \|X)=0$; and $Y_1,\dots,Y_k$ are
\emph{completely conditionally independent over $X$} if $Y_I$ and $Y_J$ are
conditionally independent over $X$ for arbitrary disjoint subsets $I$ and
$J$ of the indices. An equivalent condition is
$$
   f(Y_1\cdots Y_k\|X) = f(Y_1\|X)+ \cdots + f(Y_k\|X),
$$
which similarly implies $f(Y_I\|X)=\sum_{i\in I} f(Y_i\|X)$ for every
index set $I$.


\section{The Maximum Entropy Method}\label{sec:method}

In general terms, the principle of maximum entropy is easy to formulate:
\emph{if a probability distribution is specified only partially, take the
one with the largest entropy,} see, e.g., \cite{maxentp}. In the particular
case applied here ``partial specification'' means fixing some, but not all,
marginal distributions. To be more concrete, suppose $\xi$ is distributed
jointly on the base set $N$. Partition $N$ into three non-empty subsets as
$N= Y\cups X\cups Z$. Take $n\ge 1$ disjoint copies of $Y$ and $m\ge 1$
disjoint copies of $Z$ to form the enlarged base set
$$
  N^* = Y_1\cups \cdots \cups Y_n \cups X \cups Z_1\cups \cdots\cups Z_m.
$$
Consider the collection of those distributions $\xi^*$ on $N^*$ whose
marginals on $Y_iX$ are equal to $\xi_{\YX}$, and marginals on $XZ_j$ are
equal to $\xi_{XZ}$. That is, the marginal of $\xi$ on $\YX$ and the
marginals of $\xi^*$ on all $Y_iX$ are the same as well as the marginal of
$\xi$ on $XZ$ and the marginals of $\xi^*$ on $XZ_j$. This collection of
distributions is not empty, as one can take each $Y_i$ to be the same as
$Y$, and each $Z_j$ to be the same as $Z$. The total entropy is a strictly
concave function of the probability masses, and fixing certain marginals
imposes linear constraints on those masses. Consequently, there is a unique
optimal distribution $\xi^*$ with maximum total entropy, see
\cite{convex-optim}. Although structural properties of the maximum entropy
distributions are mainly unknown, they are known to satisfy numerous
conditional independencies. For this particular case, these are stated as
Lemma \ref{lemma:maxent1} below.

\begin{lemma}\label{lemma:maxent1}
In the distribution with maximum total entropy, the subsets $Y_1,\dots,
Y_n$ and $Z_1,\dots, Z_m$ are completely conditionally independent
over $X$.
\end{lemma}

\begin{proof}
If some of the conditional independence statements do not hold,
then one can redefine the distribution keeping the specified
marginals while increasing the total entropy. For details, see
\cite{exploring}.
\end{proof}

Since identical distributions have identical entropy profiles, Lemma
\ref{lemma:maxent1} immediately implies that an entropic polymatroid
has an \emph{$n,m$-copy} as defined below:

\begin{definition}\label{def:maxent1}
Let $f$ be a polymatroid on $N$, and partition $N$ into three non-empty
subsets as $N=Y\cups X\cups Z$. Let $Y_1,\dots,Y_n$ and $Z_1,\dots,Z_m$
be disjoint copies of $Y$ and $Z$, respectively. The polymatroid $f^*$
on the base set $N^*=Y_1\dots Y_n\tsp X \tsp Z_1\dots Z_m$ is an
\emph{$n,m$-copy of $f$} if
\begin{itemize}
\item[(i)]
   $f^*$ restricted to $Y_iX$ is isomorphic to $f\restr \YX$ for
   every $i\le n$,
\item[(ii)]
   $f^*$ restricted to $XZ_j$ is isomorphic to 
   $f\restr XZ$ for every $j\le m$,
\item[(iii)]
   the $n\m+m$ subsets $Y_1, \dots, Y_n,\, Z_1, \dots, Z_m$ are completely
   conditionally independent over $X$ in $f^*$.
\end{itemize}
\end{definition}
The special version of the Maximum Entropy Method used in this paper is
based on the fact that entropic polymatroids have $n,m$-copies. For fixed
integers $n$ and $m$, polymatroids on $\YXZ$ that have an $n,m$-copy form a
polyhedral cone $\C_{n,m}$. This is proved as Claim \ref{claim:mempoly}
below. The cone $\C_{n,m}$ contains the complete entropy region $\Ga{\YXZ}$,
and is contained in the Shannon cone $\G{\YXZ}$. Consequently, bounding
facets of the cone $\C_{n,m}$ that are not facets of the Shannon cone
provide new entropy inequalities. This method is summarized as follows.

\begin{method}[special case]
Fix the base set $N$ and the partition $N=Y\cups X\cups Z$. For $n,m \ge 1$
let $\C_{n,m}$ be the polyhedral cone of those polymatroids on $N$ that have
an $n,m$-copy. Compute all bounding facets of $\C_{n,m}$ as homogeneous
linear inequalities, and delete those which are consequences of the basic
Shannon inequalities. The remaining inequalities form the maximal set of
non-Shannon inequalities provided by the partition $\YXZ$ and the numbers
$n$ and $m$.
\end{method}

Let us remark that while the maximum entropy extension is unique, the
$m,n$-copy in Definition \ref{def:maxent1} is typically not, as the
definition captures only a small part of the properties of the maximum
entropy extension. The obtained entropy inequalities form the facets of a
convex polytope; consequently, they are independent in the sense that none
of them is a consequence of the others or the Shannon~inequalities.

\smallskip
Next we prove that $\C_{n,m}$ is a polyhedral cone indeed.

\begin{claim}\label{claim:mempoly}
Polymatroids $(f,N)$ with an $n,m$-copy form a polyhedral cone.
\end{claim}

\begin{proof}
Consider the polymatroid $f$ as a $(2^{|N|}\m-1)$-dimensional vector indexed
by the non-empty subsets of $N$. Write this vector as $(\mathbf x,\mathbf
u)$ where $\mathbf x$ of dimension $d_1$ contains those coordinates where
the index $I$ is a subset of either $\YX$ or $XZ$, and $\mathbf u$ of
dimension $d_2$ contains the rest, namely those subsets that intersect both
$Y$ and $Z$. Clearly, $d_1+d_2=2^{|N|}\m-1$. Similarly, let $\mathbf y$ be
the vector formed from the values of the $n,m$-copy polymatroid $f^*$ as
indexed by the subsets of $N^*$. The vector $\mathbf y$ has dimension
$d_3=2^{|N^*|}\m-1$. Now, $(f^*,N^*)$ is a polymatroid if the vector
$\mathbf y$ satisfies all linear inequality constraints imposed by the basic
Shannon inequalities in (\bref{B1}) and (\bref{B2}); and it is an $n,m$-copy
of $f$ if, additionally, the composed vector $(\mathbf x,\mathbf y)$
satisfies the equality constraints corresponding to conditions (i)---(iii)
in Definition \ref{def:maxent1}. Consequently, there exists a matrix $M$
with $d_1+d_3$ columns, depending only on the partition $\YXZ$ and the
numbers $n$ and $m$, so that $f$ has an $n,m$-copy if and only if there is a
vector $\mathbf y$ satisfying $M\m\cdot (\mathbf x,\mathbf y)^\T\ge 0$.
Similarly, $(f,N)$ is a polymatroid if, for another matrix $B$ with
$(d_1+d_2)$ columns expressing the basic Shannon inequalities for $\YXZ$, we
have $B\m\cdot(\mathbf x, \mathbf u)^\T\ge 0$. Thus the collection of
polymatroids on $N$ that have an $n,m$-copy is the set
\begin{align*}
    \mathcal Q = \big\{(\mathbf x,\mathbf u)\in\R^{d_1+d_2} :{}&
       B\m\cdot(\mathbf x,\mathbf u)^\T\ge 0, \mbox{ and}\\
       & M\m\cdot (\mathbf x,\mathbf y)^\T\ge 0
       \mbox{ for some } \mathbf y\in\R^{d_3}
    \big\}.
\end{align*}
Here $M$ and $B$ are matrices with integer entries; these matrices depend
only on $\YXZ$, $n$, and $m$. Since $\mathcal Q$ is the intersection of a
polyhedral cone and the projection of a polyhedral cone, it is also a
polyhedral cone, as claimed.
\end{proof}

From the proof it is clear that the $\mathbf u$-part of $\mathcal Q$ is
constrained only by the basic Shannon inequalities encoded in the matrix
$B$. Furthermore, constraints on $\mathbf x$ imposed by the first condition
are contained in the second one. Thus, it suffices to consider the bounding
facets of
\begin{equation}\label{eq:qstar}
   \mathcal Q^* = \big\{ \mathbf x\in\R^{d_1}:
     M\m\cdot(\mathbf x,\mathbf y)^\T\ge 0
     \mbox{ for some } \mathbf y\in\R^{d_3}
   \big\}
\end{equation}
for new entropy inequalities. This is because, due to the duality theorem
of linear programming \cite{ziegler}, facets of $\mathcal Q$ are
convex linear combinations of facets of $\mathcal Q^*$ and facets
corresponding to the basic Shannon inequalities for the base set $\YXZ$.

Coordinates in $\mathbf x$ are indexed by subsets of $\YX$ and $XZ$, so the
inequalities provided by the bounding facets of $\mathcal Q^*$ contain only
elements of the restrictions $f\restr \YX$ and $f\restr XZ$. We emphasize
that these restrictions are not arbitrary polymatroids on $\YX$ and $XZ$
with a common restriction on $X$, as they also have a common extension,
namely $f$. Conditions ensuring the existence of such a common extension are
assumed to hold, see \cite{Csirmaz.oneadhesive}, and they do not contribute
towards the non-Shannon entropy inequalities we are searching for.


\section{What to Compute? How to Compute?}\label{sec:prepare}

As discussed in Section \ref{sec:method}, the task of finding new
non-Shannon entropy inequalities implied by the existence of an $n,m$-copy
reduces to enumerating all facets of the polyhedral cone $\mathcal Q^*$
defined in (\ref{eq:qstar}). However, without further reduction, this
polyhedral computation is intractable even for small parameter values.
Therefore, in this Section we look at some general methods to reduce the
complexity of the computation, and then discuss how the number of elements
in the $\YXZ$ partition was chosen.

\subsection{Tight and Modular Parts}

Both the polyhedral region $\G N$ and the closure of the entropy region
$\clGa N$ decompose naturally into direct sums of \emph{modular} and
\emph{tight} parts, see \cite{entreg}. To discuss this result, let us first
introduce some notation. For $i\in N$ define the function $\r_i$ on the
non-empty subsets $A$ of $N$ as
\begin{equation}\label{eq:r}
   \r_i: A \mapsto \begin{cases}
      1 & \mbox{ if } i\in A,\\
      0 & \mbox{ otherwise. }
   \end{cases}
\end{equation}
Non-negative multiples of $\r_i$ are clearly entropic polymatroids;
\emph{modular polymatroids} are, by definition, the conic combinations of
the vectors $\{\r_i:i\in N\}$. For a polymatroid $(f,N)$, a singleton $i\in
N$ and a real number $\alpha\ge 0$, the function $f\dn^i_\alpha$ is defined
on the non-empty subsets of $N$ as follows:
$$
  f\dn^i_\alpha : A\mapsto \min\{f(Ai)\m-\alpha, ~f(A)\}.
$$
When $\alpha$ is set to $f(i\|N\sm i)$, $f\dn^i_\alpha$ is denoted simply by
$f\dn^i$. Note that for $i\notin A$ we have $f(Ai)-f(A)\ge f(N)-f(N\sm i)
=f(i\|N\sm i)$ by submodularity. Consequently, $f\dn^i$ can be written
explicitly as
$$
   f\dn^i (A) = \begin{cases}
      f(A)-f(i\|N\sm i) & \mbox{ if } i \in A, \\
      f(A)              & \mbox{ if } i\notin A.
   \end{cases}
$$
Therefore, $f= f\dn^i + f(i\|N\sm i)\tsp\r_i$, where $\r_i$ is the
polymatroid defined in (\ref{eq:r}). The result of \emph{tightening $f$ at
$i$} is the function $f\dn^i$. The \emph{tight part of $f$}, denoted by
$f\dn$, is the result of tightening $f$ at every element of its base set
$N=\{i_1,\dots,i_n\}$:
$$
    f\dn = (\cdots(f\dn^{i_1})\dn^{i_2} \cdots )\dn^{i_n}.
$$
This result is independent of the order in which the reductions are applied,
which is also shown by the decomposition formula
$$
   f = f\dn + \sum_{i\in N} f(i\|N\sm i)\tsp \r_i .
$$
The proof of the following lemma can be found in \cite{exploring} or
\cite{M.twocon}. In this paper only the first part of the lemma is needed,
which can be verified by direct computation.

\begin{lemma}\label{lemma:down}
Let $0\le \alpha\le f(i)$. 
If $f$ is a polymatroid, then $f\dn^i_\alpha$ is also a polymatroid.
If, in addition, $f$ is almost entropic, then so is $f\dn^i_\alpha$.
\qed
\end{lemma}

Accordingly, $f\dn$ (the \emph{tight part} of $f$) is a polymatroid, and it
is also almost entropic (aent) whenever $f$ is aent. The difference $f-f\dn$
is the \emph{modular part}, and it is a modular polymatroid. This
decomposition of $f$ into a tight and a modular part is unique, and both
parts are aent if $f$ is aent.

The cone formed by the modular polymatroids over $N$ is $|N|$-dimensional,
and is generated by the linearly independent vectors $\{\r_i:i\in N\}$. The
cone of tight polymatroids is orthogonal to this (modular) cone, and so to
every vector $\r_i$; and is bounded by the hyperplanes corresponding to the
basic Shannon inequalities in (\bref{B2}). The cone of tight, almost
entropic polymatroids is similarly orthogonal to the modular cone. A
consequence of this decomposition is that linear bounds on the entropic cone
also decompose into bounds on the tight part and bounds on the modular
part---the latter being trivial, that is, a Shannon inequality. The normal
$\mathbf n$ of a supporting hyperplane of the tight part is necessarily
orthogonal to all vectors $\r_i$, that is, the scalar products $\mathbf
n\cdot \r_i$ are zero. Consequently, if the normal has the coordinates
$\mathbf n=\langle t_I:I\subseteq N\rangle$, then the sum $\sum \{ t_I: i\in
I\}$ is zero for every $i\in N$. For this reason, these hyperplanes are
called \emph{balanced}. The tight component of any entropy inequality is
balanced, and it is also an entropy inequality. This fact is equivalent to
saying that every entropy inequality can be strengthened to become a
balanced one, see \cite{balanced}.

\smallskip
From the above it follows that the facets of the cone $\mathcal
Q^*$ belong to two disjoint groups. There are $|N|$ (trivial, Shannon)
facets that bound the modular part of $\mathcal Q^*$, and the rest
bound the tight part. The normal vectors of the facets in the second
group are balanced, and only they can provide non-Shannon inequalities.
Therefore, it suffices to consider only the tight part of $\mathcal Q^*$.
This part is generated by a smaller collection of polymatroids, has
fewer dimensions, and so can be handled more efficiently.

\begin{claim}\label{claim:reduction1}
The tight part of $\mathcal Q^*$ is generated by the $n,m$-copies of the
polymatroids $f$ on $N=\YXZ$ that are
{\upshape(i)} tight;
{\upshape(ii)} satisfy $f(Y,Z\|X)=0$; and
{\upshape(iii)} for all $y\in Y$, $f(y\|\YX\sm y)=0$, and for all
$z\in Z$, $f(z\|XZ\sm z)=0$.
\end{claim}

Observe that the tightness of $f$ at the elements of $Y$ and $Z$ follows
from condition (iii) and submodularity; thus (i) is relevant only for
elements of $X$.

\begin{proof}
Let $f^*$ be an $n,m$-copy of $f$. In the definition of $\mathcal Q^*$
only the values of $f\restr \YX$ and the values of $f\restr XZ$ are
used. Therefore, $f$ can be replaced with any other polymatroid that
has the same restrictions. Such a polymatroid is $f^*\restr Y_1XZ_1$
by part (i) of Definition \ref{def:maxent1}, which gives (ii). For (iii)
let $y\in Y$, and $\alpha=f(y\|\YX\sm y)$. Apply Lemma \ref{lemma:down}
to $f^*$ and all instances of $y$ in the copies $Y_i$ to get the new
polymatroid $g^*$. Denoting the instance of $y$ in $Y_1$ by $y_1$,
the lemma provides $g^*(y_1\|Y_1X\sm y_1)=0$. In addition, $g^*$ is an
$n,m$-copy of its restriction to $Y_1XZ_1$. Since this restriction
and the polymatroid $f$ differ only by a modular shift on subsets of
$\YX$ and $XZ$, their tight parts are the same. A similar reduction on
elements of $Z$, and finally on elements of $X$ provide the statement.
\end{proof}

Using Claim \ref{claim:reduction1}, the number of columns in the constraint
matrix $M$ in (\ref{eq:qstar}) can be significantly reduced. It is so since,
by the tightness of $f$, $f^*(Ai)=f^*(A)$ holds for many subsets $A$ of
$N^*$ with few elements, and this equality implies $f^*(Bi)= f^*(B)$ for
every $A\subseteq B\subset N^*$.

\subsection{Symmetry}

The inherent symmetry in the $n,m$-copy allows for another significant
complexity reduction. Let $\pi$ be one of the $(n!\tsp m!)$ permutations of
the base set $N^*$ that permutes the subsets $Y_i$ and the subsets $Z_j$
independently. This permutation naturally extends to the subsets of $N^*$,
and then to the polymatroids on $N^*$. The $n,m$-copy $f^*$ of $f$ is
\emph{symmetric} if it is invariant for each such permutation $\pi$, that
is, $f^*(A)=(\pi f^*)(A) = f^*(\pi A)$ for all $A\subseteq N^*$.

\begin{claim}\label{claim:reduction2}
$f$ has an $n,m$-copy if and only if it has a symmetric $n,m$-copy.
\end{claim}
\begin{proof}

If $f^*$ is an $n,m$-copy of $f$, then clearly so is $\pi f^*$. Since
conditions (ii) and (iii) in Definition \ref{def:maxent1} are linear,
they are also satisfied by the average of all such permutations of $f^*$,
that is, by the polymatroid $g^*=(n!\tsp m!)^{-1}\sum_\pi \pi f^*$.
Clearly, $g^*$ is a symmetric $n,m$-copy of $f$.
\end{proof}

Symmetry alone reduces the number of auxiliary variables in the definition
of $\mathcal Q^*$ from exponential in $n$ and $m$ to polynomial in these
parameters.

\subsection{No New Inequality}

In some cases, the computations required by the Maximum Entropy Method of
Section \ref{sec:method} can be simplified further, or even completely
avoided. The first claim of this subsection states that certain polymatroids
do not contribute to new entropy inequalities.

\begin{claim}\label{claim:modularX}
Suppose $f$ is a polymatroid on $N=\YXZ$, and $f$ restricted to $X$ is
modular. Then $f$ has an $n,m$-copy for every $n$ and $m$.
\end{claim}
\begin{proof}
The statement follows from the following lemma by induction.
\end{proof}

\begin{lemma}\label{lemma:minconstr}
Suppose that the polymatroids $(f_1,\YX)$ and $(f_2,XZ)$ have a common
restriction on $X$ which is modular. Then there is a polymatroid
$(g,\YXZ)$ that extends both $f_1$ and $f_2$ such that $g(Y,Z\|X)=0$.
\end{lemma}

\begin{proof}
For $I\subseteq Y$, $J\subseteq X$ and $K\subseteq Z$ define
$$
  g(IJK) \eqdef \min_L ~ \{ 
      f_1(IL)+ f_2(LK) - f_1(L) : ~ J\subseteq L\subseteq X \}.
$$
Using the fact that $f_1\restr X$ and $f_2\restr X$ are isomorphic
and modular, a simple calculation shows that $g$ is a polymatroid
and satisfies the requirements. For details, consult
\cite{Csirmaz.oneadhesive,exploring} or \cite{M.fmadhe}.
\end{proof}

If either $Y$ or $Z$ has a single element, then one does not need to
look beyond $n,1$-copies.

\begin{claim}\label{claim:singleZ}
Suppose $|Z|=1$. Entropy inequalities generated by $n,m$-copies of
polymatroids on $\YXZ$ are also generated by $n,1$-copies.
\end{claim}
\begin{proof}
We claim that the cone generated by the tight part of $n,m$-copies is the
same as the cone generated by the $n,1$-copies. To prove this, let $f$ be a
polymatroid on $\YXZ$ that satisfies the conditions of Claim
\ref{claim:reduction1}, and let $f^*$ be an $n,1$-copy of $f$ so that $f$ is
identified with $f^*\restr Y_1XZ_1$ where $Z_1$ has a single element $z$.
Let $g^*$ be the polymatroid when $m\m-1$ identical copies of $z$ are added
to $f^*$. We claim that $g^*$ is an $n,m$-copy. The only non-trivially
satisfied condition is that the copies of $z$, $z_1$ and $z_2$, are
independent over $X$. Since $f(z\|X)=0$ by (iii) of Claim
\ref{claim:reduction1}, we have $g^*(Xz_1)=g^*(Xz_2)=g^*(Xz_1z_2)= g^*(X)$,
thus $g^*(z_1,z_2\|X)=0$. Since $g^*\restr Y_1XZ_1$ and $f^*\restr Y_1XZ_1$
are the same polymatroids, the $n,m$-cone is part of the $n,1$-cone, as
claimed.
\end{proof}

\subsection{Problem Parameters}

By Claim \ref{claim:modularX}, the Maximum Entropy Method does not yield new
inequalities when $f\restr X$ is modular. This is certainly the case when
$|X|=1$, so we must have $|X|\ge 2$. By Claim \ref{claim:singleZ}, if
$|Y|=|Z|=1$, then beyond the $1,1$-copy, no additional inequalities are
generated. The smallest parameter setting when new entropy inequalities are
expected as the number of copies grows is $|Y|=2$, $|X|=2$, and $|Z|=1$. We
fix these sizes, as well as the labels of the members of each set as
$$
   X=\{a,b\}, ~~ Y=\{c,d\}, ~\mbox{ and  } Z=\{z\}.
$$
Since $|Z|=1$, according to Claim \ref{claim:singleZ}, it suffices to
consider $n,1$-copies only. To simplify the notation, the extra $1$ will be
dropped and we write $n$-copy instead. We also explicitly state the
definition of the $n$-copy for this particular partition.

\begin{definition}\label{def:n-copy}
Let $f$ be a polymatroid on $N=\{abcdz\}$, and let $n\ge 1$. The
polymatroid $f^*$ on the base set $N^*=abz\cup\{c_id_i:1\le i\le n\}$
is an \emph{$n$-copy of $f$}, if
\begin{itemize}
\item[(i)] $f^*\restr abz$ is isomorphic to $ f\restr abz$,
 and, for each $i\le n$, with the $c_i\leftrightarrow c$, $d_i
 \leftrightarrow d$ correspondences,  $f^*\restr abc_id_i$ is
 isomorphic to $f\restr abcd$; 
\item[(ii)]
   $\{c_id_i:i\le n\}$ and $z$ are completely conditionally 
   independent over $ab$.
\end{itemize}
\end{definition}

This special case of the Maximum Entropy Method provides new non-Shannon
entropy inequalities based on the fact that entropic polymatroids on the
$5$-element base set $abcdz$ have an $n$-copy for each $n\ge 1$. The steps
we will follow are as below:
\begin{enumerate}
\item 
Fix the number of copies $n$, called a \emph{generation}. Determine the
generating matrix
$M$ of the cone $\mathcal Q^*$ as specified in Claim
\ref{claim:mempoly} using only polymatroids
that satisfy the conditions of Claim \ref{claim:reduction1}, and are
symmetric by Claim \ref{claim:reduction2}.
\item\blabel{(2)}
The new inequalities are provided by the non-Shannon facets of the
tight part of $\mathcal Q^*$; these facets can be computed using
some polyhedral algorithm from the generating matrix $M$.
\end{enumerate}


\section{Computation}\label{sec:compresults}

The cone $\mathcal Q^*$ whose non-Shannon bounding facets provide the new
entropy inequalities sits in the $d_1=19$-dimensional Euclidean space with
coordinates indexed by the non-empty subsets of $\YX=\{abcd\}$ and
$XZ=\{abz\}$. Fix the number of copies to $n\ge 1$. This choice also fixes
the dimension of the vector $\mathbf y$ to $d_3=8\m\cdot 4^n-1$. The
generating matrix $M$ of the polyhedral cone $\mathcal Q^*$ from
(\ref{eq:qstar}) is repeated here:
\begin{equation}\label{eq:Qrep}
   \mathcal Q^* = \big\{ \mathbf x\in\R^{d_1}:
     M\m\cdot(\mathbf x,\mathbf y)^\T\ge 0
     \mbox{ for some } \mathbf y\in\R^{d_3}
   \big\}.
\end{equation}
The modular part of $\mathcal Q^*$ is five-dimensional, and so its tight
part sits in a $14$-di\-men\-sional subspace of $\R^{19}$, meaning that the
dimension of $\mathbf x$ is reduced to $d_1=14$. By Claim
\ref{claim:reduction2}, the polymatroid $f^*$ can be assumed to be symmetric
for all permutations of the copies $Y_i=c_id_i$. This means that the value
$f^*(J)$ depends only on the cardinality of the sets
$$
   \big\{i\le n: J\cap Y_i=\{c_i\}\big\}, ~~~
   \big\{i\le n: J\cap Y_i=\{d_i\}\big\}, ~~~
   \big\{i\le n: J\cap Y_i=\{c_id_i\}\big\},
$$
and on the set $J\cap XZ$. Since $|XZ|=3$, this property alone reduces $d_3$
to $8{n+3\choose 3}-1$. Additional conditional independences provided by
Claim \ref{claim:reduction1} reduce this number further to
\begin{equation}\label{eq:d3}
   d_3 = 4{n\m+3\choose 3}+2{n\m+2\choose 2} + 1,
\end{equation}
which number was verified empirically up to $n=10$.

A structural property of the polymatroid region $\G{abcd}$ on the
four-element set $abcd$ allows us to further reduce the complexity of the
polyhedral computation required in step \bref{(2)} above. The region
$\G{abcd}$ has a central part and six permutationally equivalent
``protrusions,'' depending on the signs of the Ingleton expressions
$$
  f\abcd,~ f\acbd,~ f[adbc], ~ f[bcad\tsp], ~f[bdac], \mbox{ and } f\cdab.
$$
If all of them are non-negative, then the restriction $f\restr abcd$ is a
linear polymatroid; otherwise exactly one of these Ingleton expressions is
negative, see, e.g., \cite{Ma.Stud}. Accordingly, the cone $\mathcal Q^*$ is
cut into seven parts by these Ingleton hyperplanes: the central part where
all Ingleton values are non-negative, and six other parts where exactly one
of the expressions is negative. The facets of each part of $\mathcal Q^*$
can be computed separately.

Parts of $\mathcal Q^*$ on the negative side of $\acbd$, $[adbc]$,
$[bcad\tsp]$, and $[bdac]$ are isomorphic because swapping $a\leftrightarrow
b$ and swapping $c\leftrightarrow d$ are symmetries of $\mathcal Q^*$.
Therefore, it suffices to consider only one of them. The central part, where
every Ingleton expression is non-negative, does not yield new inequalities.
This follows from Lemma \ref{lemma:no-linear} below, as the elements of the
central part are linear. Note that the proof uses the additional fact that
$XZ$ has exactly three elements.

\begin{lemma}\label{lemma:no-linear}
If $f$ restricted to $abcd$ is linear, then $f$ has an $n$-copy for all
$n\ge 1$.
\end{lemma}

\begin{proof}
Since every polymatroid on three elements is linear, and linearly
representable polymatroids on three or four elements are representable over
any field, we can assume, after scaling and using continuity, that both
$f\restr abcd$ and $f\restr abz$ are $\mathbb F$-linearly representable over
the same finite field $\mathbb F$. Denote the two representing vector spaces
by $V^1$ and $V^2$, and consider the subspace arrangements $(V^1_a,V^1_b)$
and $(V^2_a,V^2_b)$ in the two vector spaces. Now $V^i_a$ and $V^i_b$ have
dimensions $f(a)$ and $f(b)$, respectively, and their linear span has
dimension $f(ab)$. Therefore, these arrangements are isomorphic, and $V^1$
and $V^2$ can be glued along the linear span of $(V^1_a,V^1_b)$ and
$(V^2_a,V^2_b)$. This gluing yields an $\mathbb F$-linear polymatroid $g$
that has the same restrictions on $abcd$ and on $abz$ as $f$ does. Since
this $g$ is entropic, it has an $n$-copy for every $n\ge 1$. This $n$-copy
is also an $n$-copy of $f$, as required.
\end{proof}

Consequently, up to the $a\leftrightarrow b$ and $c\leftrightarrow d$
symmetries, three mutually exclusive cases are left: $f\abcd <0$,
$f\acbd<0$, and $f\cdab<0$. Using the homogeneity of $\mathcal Q^*$, the
Ingleton value can be set to $-1$, in effect taking a cross-section of
$\mathcal Q^*$ that has one fewer dimension. Facets of the part of $\mathcal
Q^*$ we are considering are also facets of these cross-sections;
consequently, only facets of the cross-sections need to be computed. We
consider these three cases separately in the subsections below.

The definition (\ref{eq:Qrep}) of the cone $\mathcal Q^*$ uses the
$19$-dimensional coordinate system where the coordinates of the vector
$\mathbf x$ are labeled by the non-empty subsets of $abcd$ and $abz$. In all
three cases we perform calculations in different coordinate systems that are
chosen so that
\begin{itemize}
\item the first coordinate is the Ingleton expression defining the
cross-section;
\item the tight and modular parts of the cross-section have disjoint
coordinates;
\item polymatroids have non-negative coordinates apart from the 
Ingleton coordinate.
\end{itemize}
The first property allows setting the Ingleton value explicitly. Based on
the second property, the tight part of the cross-section can be separated by
dropping some coordinates; and the third property potentially reduces the
complexity of the polyhedral enumeration algorithm.

\subsection{Case I}\label{subsec:c1}\blabel{subsec:c1}

The cone $\mathcal Q^*$ is intersected with the hyperplane $\abcd=-1$.
In this case we use the coordinate system
$$\mathcode`\:="3A\setlength\arraycolsep{2.5pt}%
\begin{array}{llllll}
 \Co_1: & \abcd, \\[3pt]
 \Co_2\mbox{--}\Co_4: & (a,b\|c), & (a,c\|b), & (b,c\|a), \\[3pt]
 \Co_5\mbox{--}\Co_7: & (a,b\|d), & (a,d\|b), & (b,d\|a), \\[3pt]
 \Co_8\mbox{--}\Co_{11}: & (c,d\|a), & (c,d|b),& (c,d), & (a,b\|cd), \\[3pt]
 \Co_{12}\mbox{--}\Co_{14}:& (a,b\|z), & (a,z\|b), & (b,z\|a), \\[3pt]
 \Co_{15}\mbox{--}\Co_{19}: & (a\|bcd), & (b\|acd), & (c\|abd), 
   & (d\|abc), & (z\|ab).
\end{array}$$
Coordinates $\Co_{15}\mbox{--}\Co_{19}$ cover the modular part of $\mathcal
Q^*$. The tight part is spanned by the coordinate vectors
$\Co_1\mbox{--}\Co_{14}$, and each of these vectors is orthogonal to the
modular part. Let $\tilde P_1$ be the inverse of the matrix of this
coordinate transformation, and the vector $\mathbf p_1$ be the first row of
$\tilde P_1$. Let $P_1$ be the submatrix formed from rows $2$ to $14$ of
$\tilde P$. Coordinates of the vector $\mathbf x\in\R^{19}$ in this
coordinate system are $\tilde P_1\tsp \mathbf x^\T$, and, in particular, the
Ingleton value $f\abcd$ is the scalar product $\mathbf p_1\m\cdot \mathbf
x$. Consequently, the tight part of the intersection of $\mathcal Q^*$ and
the hyperplane $\abcd=-1$ in this coordinate system is
\begin{equation}\label{eq:Q1}
   \mathcal Q^*_1 = \big\{ P_1\mathbf x^\T:
            \mathbf p_1\m\cdot\mathbf x=-1, \mbox{ and }
            M\m\cdot (\mathbf x,\mathbf y)^\T \ge 0
            \mbox{ for some }\mathbf y\in\R^{d_3} \big\}.
\end{equation}
Finding all facets of $\mathcal Q^*_1$ determined by the matrices $M$ and
$\tilde P$ is closely related to linear multiobjective optimization
\cite{impossible}, and can benefit significantly by working in the
$13$-dimensional \emph{target} space \cite{ehrgott} instead of the
significantly larger, $d_3$-dimensional \emph{problem} space. We have
developed a variant of Benson's inner approximation algorithm
\cite{inner,bensolve} which takes advantage of the additional special
property that $\mathcal Q^*_1$ is in the non-negative orthant of the target
space. The program is available on GitHub as
\url{https://github.com/csirmaz/information-inequalities-5}.

\begin{table}[h!bt]
\def\bb{\kern -5pt}
\begin{tabular}{crrrrr}
$n$ & Rows & Columns\bb & ~~~Facets\bb & Vertices\bb & Time (s)\bb  \\
\hline
\rule{0pt}{11pt}%
1   & 76   &  23     &  16    &    19    &   0.01 \\
2   & 284  &  53     &  21    &    43    &   0.03 \\
3   & 706  & 101     &  34    &   155    &   0.35 \\
4   & 1416 & 171     &  63    &   675    &   3.54 \\
5   & 2488 & 267     & 120    &  2171    &  38.25 \\
6   & 3996 & 393     & 221    &  6275    &   5:24~\;~{\space} \\
7   & 6014 & 533     & 386    & 14523    &  36:45~\;~{\space} \\
8   & 8616 & 751     & 635    & 31379    & 2:59:17~\;~{\space} \\
9   & 11876 & 991    & 1000   & 61627    & 13:13:45~\;~{\space} \\
\hline
\end{tabular}

\vskip 5pt

\caption{Results for the case $\abcd<0$.}\label{table:2}%
\end{table}

Table \ref{table:2} shows the sizes of the generating matrix $M$, the total
number of facets and vertices (including extremal directions) of the
cross-section $\mathcal Q^*_1$, and the running time of the vertex
enumeration algorithm on a single-core desktop computer with an
Intel\textsuperscript{\textregistered} Core{\texttrademark} i5-4590 CPU @
3.30GHz processor and 8 GB of memory. The running time was taken up almost
exclusively by the underlying LP solver. The number of columns is the
dimension $d_3$ from formula (\ref{eq:d3}). While the number of facets grows
quite moderately with $n$, the number of vertices more than doubles at each
generation. The matrix $M$, despite numerous improvements, is highly
degenerate, and numerical instability, originating from both the LP solver
and the applied polyhedral algorithm, prevented the completion of the
computation for larger values of $n$. The results of the computation are
presented in Section \ref{sec:ineq}.

\subsection{Case II}

The cone $\mathcal Q^*$ is intersected with the hyperplane $\acbd=-1$. The
coordinate system is similar to the one used in Section \ref{subsec:c1}.
Base elements $b$ and $c$ are swapped in coordinates $\Co_{1}$ to
$\Co_{11}$, while the other coordinates remain unchanged. The tight part of
the intersection, denoted by $\mathcal Q^*_2$, is defined similarly with the
same matrix $M$ but a different coordinate transformation matrix $\tilde
P_2$, vector $\mathbf p_2$, and submatrix $P_2$ as
\begin{equation}\label{eq:Q2}
   \mathcal Q^*_2 = \big\{ P_2\mathbf x^\T:
            \mathbf p_2\m\cdot\mathbf x=-1, \mbox{ and }
            M\m\cdot (\mathbf x,\mathbf y)^\T \ge 0
            \mbox{ for some }\mathbf y\in\R^{d_3} \big\}.
\end{equation}
The problem size, number of facets and vertices, and the running time in
seconds are summarized in Table \ref{table:9}.
\begin{table}[h!bt]
\def\bb{\kern -5pt}
\begin{tabular}{crrrrr}
$n$ & Rows & Columns\bb & ~~~Facets\bb & Vertices\bb & Time (s)\bb  \\
\hline
\rule{0pt}{11pt}%
1   & 76   &  23     &  16    &    19    &  0.00 \\
2   & 284  &  53     &  18    &    25    &  0.03 \\
3   & 706  & 101     &  20    &    35    &  0.14 \\
4   & 1416 & 171     &  22    &    49    &  0.65 \\
5   & 2488 & 267     &  24    &    67    &  2.36 \\
6   & 3996 & 393     &  26    &    89    &  7.37 \\
7   & 6014 & 533     &  28    &   115    & 32.21 \\
8   & 8616 & 751     &  30    &   145    & 1:12~\;~{\space} \\
9   & 11876 & 991    &  32    &   179    & 5:01~\;~{\space} \\
\hline
\end{tabular}

\vskip 5pt

\caption{Results for the case $\acbd<0$.}\label{table:9}
\end{table}
Both the number of facets and the number of vertices grow moderately. A
plausible conjecture is that, in general, the number of facets is $2n+14$,
and the number of vertices is $2n^2+17$.

The running time is significantly shorter than in Section \ref{subsec:c1}.
It is explained by the fact that the polyhedral algorithm requires solving
an LP instance for each vertex and each facet in the result, and those
numbers are significantly smaller here. The generating matrix $M$ is the
same in both cases, implying that the problem size is the same. Numerical
instability prevented completing the computation for $n=10$ even in this
case.

\subsection{Case III}

No new inequality is generated when the cone $\mathcal Q^*$ is intersected
with the hyperplane $\cdab=-1$. This can be proved as follows. Since this
intersection, denoted by $\mathcal Q^*_3$, is an (unbounded) polyhedron,
every polymatroid in $\mathcal Q^*_3$ is a conic combination of its vertices
and extremal directions. These vertices and extremal directions can be
represented by certain extremal polymatroids. Conic combinations of
polymatroids that have an $n$-copy also have an $n$-copy. Consequently, it
suffices to show that these extremal polymatroids have an $n$-copy for all
$n\ge 1$.

Changing the first $11$ coordinates of the coordinate system used in Section
\ref{subsec:c1} to
$$\mathcode`\:="3A\setlength\arraycolsep{2.5pt}%
\begin{array}{llllll}
 \Co_1: & \cdab, \\[3pt]
 \Co_2\mbox{--}\Co_4: & (c,d\|a, & (a,c\|d), & (a,d\|c), \\[3pt]
 \Co_5\mbox{--}\Co_7: & (b,c\|d), & (c,d\|b), & (b,d\|c), \\[3pt]
 \Co_8\mbox{--}\Co_{11}: & (a,b\|c), & (a,b|d),& (a,b), & (c,d\|ab),
\end{array}$$
and keeping the rest, the vertex enumeration algorithm used in the previous
cases generated the vertices and extremal directions of the $13$-dimensional
tight part of $\mathcal Q^*_3$. The computation showed that it is a pointed
cone with a single vertex that has coordinates $\Co_2\mbox{--}\Co_{14}$
equal to zero (while $\Co_1=-1$) and has $14$ extremal directions, $12$ of
which are coordinate axes. Polymatroids representing the extremal directions
are linear when restricted to the base set $abcd$ (they satisfy $f\cdab=0$;
therefore, the other Ingleton values are also non-negative). Consequently,
these polymatroids have an $n$-copy for all $n\ge 1$. Finally, the remaining
polymatroid at the single vertex has $f(a,b)=0$ (as the coordinate
$\Co_{10}$ is zero), which means that $f\restr ab$ is modular. By Claim
\ref{claim:modularX} it also has an $n$-copy for all $n\ge 1$. This
concludes the proof that no non-Shannon inequality is generated in this
case.


\section{Experimental Information Inequalities}\label{sec:ineq}

For a fixed $n\ge 1$, the problem of extracting the set of non-Shannon
inequalities that form the necessary and sufficient conditions for the
existence of an $n$-copy of a polymatroid on the base set $abcdz$ was shown
to be equivalent to determining all facets of a $14$-dimensional polyhedral
cone. The cone was cut into several pieces and the facets of each piece were
computed for $n\le 9$. In this section we take a quick look at the
computational results. In the description the symbols $\Z$, $\C$, $\D$
denote the following entropy expressions:
\begin{align*}
\Z &\eqdef (a,z\|b)+(b,z\|a), \\
\C &\eqdef (a,c\|b)+(b,c\|a), \\
\D &\eqdef (a,d\|b)+(b,d\|a).
\end{align*}

\subsection{Case I}\label{subsec:resc1}\blabel{subsec:resc1}

In the $\abcd<0$ case, facets of the polyhedron $\mathcal Q^*_1$ from
(\ref{eq:Q1}) include all $13$ coordinate planes orthogonal to the
coordinate axes $\Co_2$--$\Co_{14}$. These facets correspond to the
non-negativity of the expression defining that coordinate. $\mathcal Q^*_1$
has two additional Shannon facets, corresponding to the Shannon inequalities
$(a,z)\ge 0$ and $(b,z)\ge 0$. The remaining facets determine the
non-Shannon inequalities we are interested in. They come in three flavors:
\begin{align}
     (a,b\|z) + \alpha_s \abcd+\alpha_s\Z + \beta_s\C + \gamma_s\D &\ge 0,
     \label{eq:fl1}\\[2pt]
    (a,b\|c) + \alpha_s \abcd+ \,(\alpha_s\m+\beta_s)\C \tsp\,+ \gamma_s\D &\ge 0,
    \label{eq:fl2}\\[2pt]
   (a,b\|d)+  \alpha_s \abcd+\, \beta_s\C \tsp\,+ (\alpha_s\m+\gamma_s)\D &\ge 0,
   \label{eq:fl3}
\end{align}
where $\<\alpha_s,\beta_s,\gamma_s\>$ are certain triplets of non-negative
integers. For illustration, we consider the $n=3$ case. As reported in Table
\ref{table:2}, for $n=3$ the polyhedron $Q^*_1$ has $34$ facets. These
facets determine $15$ Shannon inequalities (among which there are $13$
coordinate planes), $11$ inequalities of the form (\ref{eq:fl1}), and
$4$--$4$ inequalities of the form (\ref{eq:fl2}) and (\ref{eq:fl3}). The
$\<\alpha,\beta,\gamma\>$ triplets appearing in (\ref{eq:fl1}) are listed in
three columns in Table \ref{table:sample3}. Inequalities in (\ref{eq:fl2})
and (\ref{eq:fl3}) use triplets from the first column only; these are the
triplets that also appear in the $n=2$ generation.

\begin{table}[h!tb]\centering\begin{tabular}{c@{\qquad}c@{\qquad}c}
$\alpha,\beta,\gamma$ &
$\alpha,\beta,\gamma$ &
$\alpha,\beta,\gamma$ \\
\hline\rule{0pt}{10.5pt}%
$\<1,0,0\>$  & $\<3,0,3\>$ & $\<6,3,5\>$ \\
$\<2,0,1\>$  & $\<3,3,0\>$ & $\<6,5,3\>$ \\
$\<2,1,0\>$  & $\<4,1,3\>$ & $\<7,5,5\>$ \\
$\<3,1,1\>$  & $\<4,3.1\>$ &           \\
\hline
\end{tabular}
\vskip 5pt
\caption{The $\alpha,\beta,\gamma$ coefficient values for $n=3$}\label{table:sample3}
\end{table}

In general, inequalities in (\ref{eq:fl2}) and (\ref{eq:fl3}) are
consequences of (\ref{eq:fl1}) by replacing $z$ with $c$ and $d$,
respectively. Since the copy $c_n$ of $c$ in an $n$-copy polymatroid $f^*$
can be considered to be the variable $z$ in the $n\m-1$-copy when $f^*$ is
restricted to $N^*\sm\{d_nz\}$, inequalities valid for $n\m-1$-copy
instances must hold in an $n$-copy with $z$ replaced by $c$, and, similarly,
when $z$ is replaced by $d$. This property is confirmed by the computational
results. Additionally, all inequalities not containing the variable $z$
proved to be derivatives from the previous generation via the above
substitutions. The main goal of Sections \ref{sec:all} and
\ref{sec:reducedset} is to obtain a general description of the triplets
$\<\alpha_s,\beta_s,\gamma_s\>$ occurring in (\ref{eq:fl1}).

\subsection{Case II}\label{subsec:resc2}

Inequalities in the $\acbd<0$ case have a similar but significantly simpler
structure. Facets of the $n$-copy cone $\mathcal Q^*_2$ from (\ref{eq:Q2})
include $13$ coordinate planes, the two Shannon facets $(a,z)\ge0$ and
$(b,z)\ge 0$ as above, and additional facets corresponding to the
inequalities
\begin{align*}
   (a,b\|z) +  k\tsp\acbd + k\tsp\Z + \frac{k(k\m-1)} 2 \C &\ge 0
    ~~~~~\mbox{ where } 1\le k \le n, \\[2pt]
 (a,b\|c) + k\tsp\acbd +  \frac{(k\m+1)k} 2 \C &\ge 0
    ~~~~~\mbox { where } 1\le k < n.
\end{align*}
As noted in the $\abcd<0$ case, inequalities in the second set are instances
of ones from the first from the previous generation when $z$ is replaced by
$c$. When $z$ is replaced by $d$, the resulting inequality
$$
   (a,b\|d)+k\tsp\acbd+k\tsp\D + \frac{k(k\m-1)}2 \C \ge 0
$$
is Shannon as $\acbd + \D= (a,d\|c)+2(b,d\|a)+(a,c\|b)\ge 0$ holds in every
polymatroid.


\section{New Inequalities}\label{sec:all}

In this section we define a set of $\<\alpha,\beta,\gamma\>$ triplets, and
\emph{prove} that each of them gives rise to a non-Shannon inequality that
must hold in polymatroids having an $n$-copy. These inequalities cover those
that were discovered experimentally for $n\le 9$. We \emph{conjecture} this
set to be complete, that is, the applied MEM method yields no additional
non-Shannon inequalities; or in other words, if a polymatroid on $5$
elements satisfies all these inequalities, then it has an $n$-copy for all
$n$.

\subsection{Case I}

For notational convenience $\b(x,y)$, for \emph{binomial}, denotes
the function defined on $\N\times\N$ that satisfies the following
recurrent definition for positive integers $x$ and $y$:
\begin{align*}
  \b(0,0)&=\b(x,0)=\b(0,y)=1, \mbox{ and}\\
  \b(x,y)&=\b(x\m-1,y)+\b(x,y\m-1).
\end{align*}
Clearly, $\b(x,y)=\b(y,x)={x+y\choose x}$. Table \ref{table:b} illustrates
the initial values of $\b(x,y)$.

\begin{table}[ht]
\begin{tikzpicture}[scale=0.6]
\foreach\x in {0,2,4,6}{
  \fill[color=gray!12,opacity=0.5] (\x,-1) rectangle (\x+1,4.3);
}
\foreach\x in {0,2}{
  \fill[color=gray!12,opacity=0.5] (-1,\x) rectangle (8.4,\x+1);
}
\fill[color=gray!12,opacity=0.5] (-1,4) rectangle (8.2,4.2);
\fill[color=gray!12,opacity=0.5] (8,-1) rectangle (8.2,4.2);
\fill[color=blue!17] (4.2,-0.8) rectangle (4.8,-0.1);
\fill[color=blue!17] (-0.8,2.2) rectangle (-0.1,2.9);
\fill[color=blue!17] (4.1,2.2) rectangle (4.88,2.9);
\draw[->] (-1,0)--(8.4,0);
\draw[->] (0,-1)--(0,4.4);
\foreach\x in {0,1,2,3}{
 \draw (-0.5,\x+0.55) node{$\scriptstyle\x$};
}
\foreach\x in {0,1,2,3,4,5,6,7}{
 \draw (\x+0.5,-0.5) node{$\scriptstyle\x$};
}
\foreach\x/\y in {0/1,1/4,2/10,3/20,4/35,5/56,6/84,7/120}{
 \draw (\x+0.5,3.55) node{\small$\y$};
}
\foreach\x/\y in {0/1,1/3,2/6,3/10,4/15,5/21,6/28,7/36}{
 \draw (\x+0.5,2.55) node{\small$\y$};
}
\foreach\x/\y in {0/1,1/2,2/3,3/4,4/5,5/6,6/7,7/8}{
 \draw (\x+0.5,1.55) node{\small$\y$};
}
\foreach\x/\y in {0/1,1/1,2/1,3/1,4/1,5/1,6/1,7/1}{
 \draw (\x+0.5,0.55) node{\small$\y$};
}
\draw (-0.3,4.25) node{$y$};
\draw (8.5,-0.3) node{$x$};
\end{tikzpicture}
\caption{The function $\b(x,y)$}\label{table:b}
\end{table}

The following summation formulas will be used later.

\begin{lemma}\label{lemma:bsum}
For $x, y\in\N$ the following summation formulas hold:
\begin{align*}
  \sum_{i\le x,\;j\le y} \b(i,j) & = \b(x\m+1,y\m+1)-1, \\
  \sum_{i\le x,\;j\le y} i\tsp\b(i,j) &= x\tsp\b(x\m+1,y\m+1)
            -\b(x,y\m+2)+1.
\end{align*}
\end{lemma}
\begin{proof}
Induction on $x$ shows that $\sum_{i\le x} \b(i,y)=
\b(x,y\m+1)$, and also that
$$
    \tsum_{i\le x} i\tsp\b(i,y) = x\tsp\b(x,y\m+1)-\b(x\m-1,y\m+2).
$$
Following this, induction on $y$ gives the desired results.
\end{proof}

\begin{definition}
The set $s\subset \N\times \N$ of pairs of non-negative integers is
\emph{downward closed} if $(i,j)\in s$ implies $(i',j')\in s$ for
every non-negative $i'\le i$ and $j'\le j$. For $n\ge 1$ the
\emph{diagonal set $D_n$} is
$$
    D_n \eqdef \{ (i,j)\in\N\times\N: i+j< n \}.
$$
\end{definition}
Clearly, $D_n$ is downward closed.

\begin{definition}\label{def:triplets}
For a finite, downward closed set $s\subset\N\times\N$, define the
three-dimensional vector $\mathbf v_s$ as
$$
   \mathbf v_s =\< \alpha_s,\beta_s,\gamma_s\>
     = \sum_{(i,j)\in s} \b(i,j)\,\< 1,i,j\>,
$$
When $s$ is the empty set, define $\mathbf v_\emptyset =\< 0,0,0\>$.
\end{definition}

For illustration, the $13$ downward-closed subsets of the diagonal $D_3$, as
well as the corresponding $\<\alpha_s,\beta_s,\gamma_s\>$ triplets are
presented in Table \ref{table:dcsD3}.
\begin{table}[t]
$$
\begin{tikzpicture}[scale=1.3]
\def\DDo#1{\draw[color=gray] (#1) circle(1.5pt);}
\def\DDi#1{\draw[fill=orange] (#1) circle(1.5pt); }
\def\pxxts#1#2#3#4#5#6#7#8{%
\csname DD#1\endcsname{#7,0}%
\csname DD#2\endcsname{#7,5pt}%
\csname DD#3\endcsname{#7,10pt}%
\csname DD#4\endcsname{#7+5pt,0}%
\csname DD#5\endcsname{#7+5pt,5pt}%
\csname DD#6\endcsname{#7+10pt,0}%
\draw[color=black!30!blue] (#7-5pt,21pt) node[right] {$\scriptstyle\##8$};
}
\foreach\x in {0,2,4,6,8,10,12}{
  \filldraw[color=black!5] (\x*20pt-4pt,-6pt) rectangle (\x*20pt+13pt,28pt);
}
\draw[line width=2pt,color=black!5](-4pt,-6pt)--(240pt,-6pt)
                                    (-4pt,28pt)--(240pt,28pt);
\pxxts iooooo{0pt}1
\pxxts iooioo{20pt}2
\pxxts iooioi{40pt}3
\pxxts iioooo{60pt}4
\pxxts iioioo{80pt}5
\pxxts iioioi{100pt}6
\pxxts iioiio{120pt}{7}
\pxxts iioiii{140pt}8
\pxxts iiiooo{160pt}9
\pxxts iiiioo{180pt}{10}
\pxxts iiiioi{200pt}{11}
\pxxts iiiiio{220pt}{12}
\pxxts iiiiii{240pt}{13}
\end{tikzpicture}
$$
$$
\newcommand\bl[1]{{\color{black!30!blue}\small$\##1$}}
\begin{tabular}[t]{rc}
\bl 1 & $\<1,0,0\>$\\
\bl2 & $\<2,1,0\>$\\
\bl3 & $\<3,3,0\>$
\end{tabular}\quad\begin{tabular}[t]{rc}
\bl4 & $\<2,0,1\>$ \\
\bl5 & $\<3,1,1\>$ \\
\bl6 & $\<4,3,1\>$
\end{tabular}\quad\begin{tabular}[t]{rc}
\bl7 & $\<5,3,3\>$\\
\bl8 & $\<6,5,3\>$ \\
\bl9 & $\<3,0,3\>$
\end{tabular}\quad\begin{tabular}[t]{rc}
\bl{10} & $\<4,1,3\>$\\
\bl{11} & $\<5,3,3\>$\\
\bl{12} & $\<6,3,5\>$ \\
\bl{13} & $\<7,5,5\>$
\end{tabular}
$$
\caption{Downward closed subsets of $D_3$ and the corresponding triplets}%
\label{table:dcsD3}
\end{table}
For example, the set $\#6$ has elements $(0,0)$, $(1,0)$, $(2,0)$, and
$(0,1)$; the associated triplet is
$$
  \b(0,0)\<1,0,0\> +
\b(1,0)\<1,1,0\>+\b(2,0)\<1,2,0\>+\b(0,1)\<1,0,1\> =\<4,3,1\>.
$$
The diagonal $D_1$, which is $\#1$ in the list, has a single point, the
origin, and the corresponding vector is $\mathbf v_{D_1}=\<1,0,0\>$. In
general, the diagonal $D_n$ has $n(n+1)/2$ points, and the vector associated
with $D_n$ is
$$
\mathbf v_{D_n} =  \< 2^n\m-1,\,(n\m-2)\tsp 2^{n-1}\m+1,\, (n\m-2)\tsp 2^{n-1}\m+1\>.
$$
The following theorem provides a family of non-Shannon inequalities that
covers all inequalities that were found experimentally in Section
\ref{subsec:resc1}.

\begin{theorem}\label{thm:mainI}
Let $f$ be a polymatroid on the base set $abcdz$ that has an $n$-copy over
the $\{cd\}\{ab\}\{z\}$ partition. Then, for every downward closed set
$s\subseteq D_n$, $f$ satisfies the inequality
\begin{equation}\label{eq:main}
  (a,b\|z) + \alpha_s\big(\abcd + \Z\big) +
           \beta_s\tsp\C + \gamma_s\D \ge 0.
\end{equation}
\end{theorem}

\begin{proof}
Let $f^*$ be an $n$-copy of $f$ on the base set $N^*=\{abz\}\cup
\{c_id_i:i\le n\}$. Using Claims \ref{claim:reduction1} and
\ref{claim:reduction2}, we can assume that
\begin{itemize}
\item[(i)] $f$ is isomorphic to $f^*\restr abc_1d_1z$, and
\item[(ii)] $f^*$ is symmetric for all $n!$ permutations of the pairs
$c_id_i$.
\end{itemize}
By (i) it suffices to show that $f^*$ satisfies all inequalities in
(\ref{eq:main}). By (ii), permutationally equivalent subsets of $N^*$ have
the same $f^*$-value. Below, $c^k$, etc., stands for $k$ elements chosen
from $c_1, \ldots, c_n$. Occasionally, $c_1$, $d_1$ will also be denoted by
$c$ and $d$, and $c^{k+1}$ will be written as $cc^k$, letting $c=c_1$ be one
of the chosen elements.

A representative element for the subset $A\subseteq N^*$ will be written as
$$
    Bc^k d^\ell (cd)^m,
$$
with $k+\ell+m\le n$, where $B$ is a (possibly empty) subset of $abz$; and
from the $c_id_i$ pairs there are $k$ that intersect $A$ in $c_i$, there are
$\ell$ pairs that intersect $A$ in $d_i$, and there are $m$ pairs that
intersect $A$ in $c_id_i$. Only non-zero exponents will be presented.

The following inequality is denoted by $\I(k,\ell)$:
\begin{align*}
   \abcd &+ \Z+k\tsp\C +\ell D \ge {}\\
    &{-(a,b\|c^kd^\ell z)} + (a,b\|c^{k+1}d^\ell z)
       + (a,b\|c^kd^{\ell+1}z).
\end{align*}
By Lemma \ref{lemma:indhyp} below this inequality holds for $f^*$ when
$k+\ell<n$. Let $s\subseteq D_n$ be a downward closed set, and consider the
following combination of the inequalities $\I(k,\ell)$:
$$
   \sum_{(k,\ell)\in s} \b(k,\ell)\,\I(k,\ell).
$$
On the left hand side we have $\alpha_s$ many copies of $\abcd$ and $\Z$,
$\beta_s$ many copies of $\C$, and $\gamma_s$ copies of $\D$. On the right
hand side the only remaining negative term is $(a,b\|z)$, all others cancel
out as $\b(k\m-1,\ell)+ \b(k,\ell\m-1)=\b(k,\ell)$. Consequently, inequality
(\ref{eq:main}) holds in $f^*$, as claimed.
\end{proof}

The rest of this section is devoted to the proof of Lemma \ref{lemma:indhyp}
stating that $\I(k,\ell)$ holds in $f^*$. We start with some simple
inequalities about the copy polymatroid $f^*$. For ease of reading, we omit
the parentheses in addition to the function $f^*$.

\begin{lemma}\label{lemma:abdz}
For non-negative integers $k$, $\ell$ with $k+\ell<n$ we have
\begin{align*}
  ac^kd^{\ell+1}z &\le adz + k(ac-a) + \ell(ad-a), \\
  bc^{k+1}d^\ell z &\le bcz + k(bc-b) + \ell(bd-b).
\end{align*}
\end{lemma}
\begin{proof}
The claims are clearly true for $k=\ell=0$. Otherwise, use induction on $k$
and $\ell$ using
\begin{align*}
   ac^{k+1}d^{\ell+1}z - ac^kd^\ell z &= acX - aX \le ac -a, \\
   ac^kd^{\ell+2}z - ac^kd^{\ell+1}z &= adY - aY \le ad -a,
\end{align*}
for some subsets $X$ and $Y$ of $N^*$. The second inequality can
be proved similarly.
\end{proof}

\begin{lemma}\label{lemma:abxx}%
{\upshape(i)}
If $k+\ell\le n$ then $ abc^kd^\ell z = abz + k(abc-ab) +
\ell(abd-ab)$.

\titem{(ii)} If $k+\ell<n$ then $ab(cd) c^kd^\ell z= abcd +
(abz-ab)+k(abc-ab)+\ell(abd-ab)$.
\end{lemma}

\begin{proof}
Both statements follow from the fact that under the given conditions
$cd$, $c^k$, $d^\ell$ and $z$ are conditionally independent over $ab$.
\end{proof}

\begin{lemma}\label{lemma:upb}%
{\upshape(i)} $((cd)c^kd^\ell z - (cd) \ge (abz-ab)+k(abc-ab)+\ell(abd-ab)$.

\titem{(ii)}
$bdc^kz -bd \ge (abz-ab) + k(abc-ab)$.

\titem{(iii)}
$a(cd)c^kz - a(cd) \ge (abz-ab) + k(abc-ab)$.
\end{lemma}
\begin{proof}
For the first inequality, $(cd)c^kd^\ell z - (cd) \ge ab(cd)c^kd^\ell z -
abcd$ by submodularity. From here, apply Lemma \ref{lemma:abxx} to get the
required inequality. The other two inequalities can be proved in a similar
way.
\end{proof}

\begin{lemma}\label{lemma:indhyp}
For non-negative integers $k$, $\ell$ with $k+\ell<n$ the inequality
$\I(k,\ell)$ holds in $f^*$.
\end{lemma}
\begin{proof}
Recall that the inequality $\I(k,\ell)$ is
\begin{align*}
   \abcd &+ \Z+k\tsp\C +\ell D \ge {}\\
    &{-(a,b\|c^kd^\ell z)} + (a,b\|c^{k+1}d^\ell z)
       + (a,b\|c^kd^{\ell+1}z).
\end{align*}
Write the right hand side as the sum $T_1+T_2+T_3+T_4$ where the four terms are
\begin{align}
  T_1 &= c^kd^\ell z - c^{k+1}d^\ell z - c^kd^{\ell+1}z, \label{eq:l1} \\
  T_2 &= abc^kd^\ell z - abc^{k+1}d^\ell z - abc^kd^{\ell+1}z, \label{eq:l2}\\
  T_3 &= ac^{k+1}d^\ell z - ac^kd^\ell z + ac^kd^{\ell+1}z, \label{eq:l3}\\
  T_4 &= bc^{k+1}d^\ell z - bc^kd^\ell z + bc^kd^{\ell+1}z. \label{eq:l4}
\end{align}
We estimate each line separately. For (\ref{eq:l1}) we have
$$
    T_1 = -(c,d\|c^kd^\ell z) - (cd) c^k d^\ell z.
$$
Here the first term is ${}\le 0$, and the second term can be bounded using
part (i) of Lemma \ref{lemma:upb}. Therefore,
\begin{equation}\label{est:l1}
    T_1 \le -cd -(abz-ab) - k(abc-ab)-\ell(abd-ab).
\end{equation}
Using part (i) of Lemma \ref{lemma:abxx}, the exact value of (\ref{eq:l2})
can be computed as
\begin{equation}\label{est:l2}
   T_2 = -abz - (k+1)(abc-ab) - (\ell+1)(abd-ab).
\end{equation}
For (\ref{eq:l3}) use $ac^{k+1}d^\ell z -
ac^kd^\ell z = acX-aX \le ac-a$ and Lemma \ref{lemma:abdz} to get
\begin{equation}\label{est:l3}
   T_3\le adz+(k+1)(ac-a)+\ell(ad-a).
\end{equation}
Finally, to estimate (\ref{eq:l4}) use the similar inequality
$bc^kd^{\ell+1} z - bc^kd^\ell z \le
bd-b$ and the second statement of Lemma
\ref{lemma:abdz} to get
\begin{equation}\label{est:l4}
   T_4\le bcz + k(bc-b)+(\ell+1)(bd-b).
\end{equation}
The sum of the right hand sides in the estimates (\ref{est:l1})--(\ref{est:l4}) is
$$
   \abcd+\Z + k\C + \ell\D -(c,z|b)-(d,z|a).
$$
This amount is $\le$ than the left hand side of $\I(k,\ell)$,
proving the lemma.
\end{proof}

\subsection{Case II}

The following theorem claims that in the case of $\acbd<0$ inequalities
experimentally found in Section \ref{subsec:resc2} indeed hold for every
$n$.

\begin{theorem}\label{thm:mainII}
Let $f$ be a polymatroid on $abcdz$ that has an $n$-copy for the
partition $\{cd\}\{ab\}\{z\}$. Then $f$ satisfies the following
inequality for every $1\le k\le n$:
$$
  (a,b\|z)+ k\tsp\acbd+k\tsp\Z+\frac{k(k-1)}2\C \ge 0.
$$
\end{theorem}

We give two proofs. The first one is similar to the proof of Theorem
\ref{thm:mainI} and uses an inequality mimicking Lemma \ref{lemma:indhyp}.
The second proof is by induction and uses a technique that also recovers
some of the inequalities covered in Theorem \ref{thm:mainI}.

\begin{proof}[Proof 1 of Theorem \ref{thm:mainII}]
The following inequality holds in $f^*$ for every $0\le k<n$:
\begin{equation}\label{eq:case2}
    \acbd + \Z + k\tsp\C \ge -(a,b\|c^kz)+(a,b\|c^{k+1}z).
\end{equation}
Summing up this inequality from zero to $k\m-1$ gives the claim of Theorem
\ref{thm:mainII}, thus it suffices to prove (\ref{eq:case2}). The natural
approach of using induction on $k$ does not work. The reason for this is
that the inequality
$$
    \C \ge (a,b\|c^kz)-2\tsp(a,b\|c^{k+1}z)+ (a,b\|c^{k+2}z),
$$
required by the induction does not hold in general. Instead, we give a more
involved reasoning, resembling the proof of Lemma \ref{lemma:indhyp}. In
(\ref{eq:case2}) write $c^{k+1}$ as $cc^k$, and let $d$ be the element that
forms a pair with this $c$. Adding $(b,d\|c^kz)+(a,c\|c^kdz)$ to the
right-hand side of (\ref{eq:case2}) and rearranging, we obtain the upper
bound
$$\begin{array}{c}
  (cdc^kz-cc^kz) +(acc^kz-ac^kz) -(abcc^kz - abc^kz) +{}\\[3pt]
  {} +bcc^kz + adc^kz   - bdc^kz-acdc^kz.
\end{array}$$
Each of the seven terms is estimated separately as follows.
\begin{align}
 cdc^kz-cc^kz &\le cd-c, \label{ee1}\\
 acc^kz-ac^kz&\le ac-a, \label{ee2}\\
 abcc^kz - abc^kz &= abc-ab,  \label{ee3} \\
 bcc^kz           &\le bcz + k(bc-b),\label{ee4} \\
 adc^kz           &\le adz + k(ac-a), \label{ee5}\\
 -bdc^kz  &\le -bd - k(abc-ab)-(abz-ab), \label{ee6}\\
 -acdc^kz &\le -acd-k(abc-ab)-(abz-ab). \label{ee7}
\end{align}
Equations (\ref{ee1}) and (\ref{ee2}) follow from submodularity. Equation
(\ref{ee3}) expresses that $c$ and $c^kz$ are independent over $ab$.
Equations (\ref{ee4}) and (\ref{ee5}) are in Lemma \ref{lemma:abdz}, while
(\ref{ee6}) and (\ref{ee7}) are from Lemma \ref{lemma:upb}. The sum of the
right-hand sides of (\ref{ee1})--(\ref{ee7}) is
$$
   \acbd+\Z+k\tsp\C-(c,z\|b)-(d,z\|a),
$$
proving (\ref{eq:case2}).
\end{proof}

To describe the technique used in the second proof of Theorem
\ref{thm:mainII}, let $\mathcal E_n$ denote the collection of all linear
five-variable inequalities that are valid in every polymatroid on $abcdz$
that has an $n$-copy. Let $f^*$ be such an $n$-copy. Having $n$ instances of
$cd$, one of the $c_id_i$ pairs can be singled out, and one of its elements
can be renamed $z'$. Restricting $f^*$ to these elements is an $n\m-1$-copy
of $abcdz'$ since $z'$ and the remaining $n-1$ pairs are independent over
$ab$. Therefore, $abcdz'$ satisfies the inequalities in $\mathcal E_{n-1}$.
Let $E(a,b,c,d,z)\in\mathcal E_{n-1}$ be such an inequality, marking the
base elements explicitly. Then we have $E(a,b,c,d,c)\in \mathcal E_n$, and
also $E(a,b,c,d,d) \in \mathcal E_n$. This fact has been observed and used
in Section \ref{sec:ineq} to explain the coefficients in the obtained
inequalities that do not contain the variable $z$.

Similarly to the above, the pairs $\{(c_iz,d_iz):i\le n\}$ are isomorphic
and are independent over the pair $(az,bz)$. Therefore, they form an
$n\m-1$-copy of the polymatroid with base elements $\{az,bz,cz,dz,cz\}$.
This means that the inequality $E(az,bz,cz,dz,cz)$ is also in $\mathcal
E_n$, and so is the inequality with $dz$ in the last position.
For non-negative integers $\alpha$ and $\beta$ denote the following
inequality by $\J(\alpha,\beta)$:
$$
   \J(\alpha,\beta) \eqdef (a,b\|z) +\alpha\tsp\acbd + \alpha\Z + \beta\tsp\C \ge 0.
$$

\begin{lemma}\label{lemma:Jalphabeta}
Suppose $\J(\alpha,\beta)\in\mathcal E_{n-1}$. Then
$\J(\alpha\m+1,\alpha\m+\beta)\in\mathcal E_n$.
\end{lemma}
\begin{proof}
The following Shannon inequalities hold in every polymatroid:
\begin{align*}
  (a,b|z)+\acbd + \Z &\ge (az,bz\|cz)-3(cd,z\|ab),\\
  \acbd + \Z &\ge [az,cz,bz,dz] -3(cd,z\|ab),\\
    (a,c\|b) &\ge (az,cz\|bz) - (c,z\|ab),\\
    (b,c\|a) &\ge (bz,cz\|az) -(c,z\|ab).
\end{align*}
Since $cd$ and $z$ are conditionally independent over $ab$ in $f^*$, the
last terms are zero. Taking the first inequality once, the second one
$\alpha$ times, and the last two $(\alpha+\beta)$ times, the sum of the
left-hand sides is $\J(\alpha\m+1,\alpha\m+\beta)$, while the right-hand
side is just $\J(\alpha,\beta)$ for the $(az,bz,cz,dz,cz)$ base. Since this
inequality is in $\mathcal E_{n-1}$ by assumption, we have
$\J(\alpha\m+1,\alpha\m+\beta)\in\mathcal E_n$, as claimed.
\end{proof}

\begin{proof}[Proof 2 of Theorem \ref{thm:mainII}]
The inequality to be proved is $\J(k,{k\choose 2})$. Use induction on $k$.
For $k=0$ it is a Shannon inequality, thus it holds in every polymatroid.
For other values of $k$, Lemma \ref{lemma:Jalphabeta} says that
$\J(k,{k\choose 2})\in\mathcal E_{n-1}$ implies $\J(k\m+1,{k+1\choose
2})\in\mathcal E_n$, concluding the induction step.
\end{proof}

Lemma \ref{lemma:Jalphabeta} remains valid if, in the definition of
$\J(\alpha,\beta)$, the Ingleton expression $\acbd$ is replaced by $\abcd$.
Consequently, the inequality
$$
   (a,b\|z)+k\tsp\abcd+k\tsp\Z+\frac{k(k-1)}2\C \ge 0
$$
is also in $\mathcal E_n$ for all $k\le n$. This inequality is covered by
Theorem \ref{thm:mainI} when the downward closed set $s$ consists of the
lattice points $(0,i)$ for $0\le i\le k$.


\section{The Minimal Set of Inequalities}\label{sec:reducedset}

Experimental results reported in Section \ref{sec:compresults} and discussed
in Section \ref{sec:ineq} provided the complete list of five-variable
non-Shannon entropy inequalities implied by the existence of an $n$-copy for
$n\le 9$. Two families of non-Shannon inequalities, generalizing the ones
found experimentally, were proven, in Theorem \ref{thm:mainI} and Theorem
\ref{thm:mainII}, respectively, to hold in every polymatroid with an
$n$-copy. We \emph{conjecture} that these families actually characterize
those five-variable polymatroids that have an $n$-copy, so no further
non-Shannon inequalities can be discovered by the version of the Maximum
Entropy Method utilized in this paper.

In the $\acbd<0$ case, the family of non-Shannon inequalities provided by
Theorem \ref{thm:mainII} matches exactly the inequalities obtained
experimentally for $n\le 9$.

In the $\abcd<0$ case, the family provided by Theorem \ref{thm:mainI} is
parametrized by the downward closed subsets $s$ of the diagonal set
$D_n\subset \N\times\N$. Unfortunately, not all of the obtained inequalities
correspond to facets of the cone $\mathcal Q^*_1$. The reason is that, while
all of them are valid entropy inequalities, some are consequences of the
others. Recall that the inequality corresponding to the triplet
$\<\alpha,\beta,\gamma\>$ is
\begin{equation}\label{eq:abc}
    (a,b\|z)+\alpha\big(\abcd + \Z\big) +\beta\tsp\C+\gamma\D \ge 0.
\end{equation}
As an example, the triplet $\<5,3,3\>$, obtained from the downward closed
subsets numbered $\#7$ and $\#11$ in Table \ref{table:dcsD3}, does not
appear in Table \ref{table:sample3}, the list of computational results. The
corresponding inequality is indeed a consequence of those obtained from the
triplets numbered $\#6$, $\#10$, and $\#13$. This is so, as $\<5,3,3\>$ is
their average:
$$
  \<5,3,3\> = \frac13\big(\<4,3,1\>+\<4,1,3\>+\<7,5,5\>\big),
$$
and the same holds for the corresponding inequalities. The main goal of this
section is to obtain a description of those downward closed subsets of $D_n$
that generate facets of $\mathcal Q^*_1$, that is, inequalities that are not
consequences of the others.

Since the inequality (\ref{eq:abc}) contains the fixed term $(a,b\|z)$, and
trivially holds true when $\abcd+\Z\ge 0$, it is a consequence of the
inequalities obtained from the triplets $\{\<\alpha_i,\beta_i,\gamma_i\>:
i\in I\}$ if and only if there is a convex linear combination
$$
    \<\alpha',\beta',\gamma'\>=\tsum_{i\in I} \lambda_i\tsp
    \<\alpha_i,\beta_i,\gamma_i\>,
    ~~~\mbox{ with } \lambda_i\ge 0, \mbox{ and } \tsum_{i\in I}\lambda_i=1,
$$
such that $\alpha\le\alpha'$, $\beta\ge\beta'$, and $\gamma \ge\gamma'$. In
this case we say that $\<\alpha,\beta,\gamma\>$ is \emph{superseded} by the
set $\{\<\alpha_i,\beta_i,\gamma_i\>:i\in I\}$. If $\mathbf v_S=\<\alpha_S,
\beta_S,\gamma_S\>$ is not superseded by other elements of this family, then
$\mathbf v_S$ is called \emph{extremal}. Actually, by the above observation,
extremal vectors are the vertices of the convex hull of the set of triplets
$\mathbf v_S$ as $S$ runs over the downward closed subset of $D_n$. By
Carath\'eodory's theorem, see \cite{ziegler}, $\mathbf v_S$ is superseded if
and only if it is (also) superseded by a set with at most three elements.

Lemma \ref{lemma:convex} below gives a necessary and sufficient condition
for the vector $\mathbf
v_S$ to be superseded by a special three-element set. For a subset $S$ of
$\N\times\N$ we write $S\m+(i,j)$ for adding the point $(i,j)$ to $S$, and
$S\m-(i,j)$ to remove $(i,j)$ from $S$. In the first case it is tacitly
assumed that $(i,j)$ is not in $S$, and in the second case that $(i,j)\in
S$.

\begin{lemma}\label{lemma:convex}
Let $i_1<i_2<i_3$, and $j_1>j_2>j_3$.
\begin{itemize}
\item[(i)] $\mathbf v_S$ is
superseded by the vectors $\{\mathbf v_{S+(i_1,j_1)},
\mathbf v_{S-(i_2,j_2)}, \mathbf v_{S+(i_3,j_3)} \}$ if and only if
$$
    \frac{~j_2-j_3~}{i_3-i_2} \ge \frac{~j_1-j_3~}{i_3-i_1}.
$$
\item[(ii)] $\mathbf v_S$ is superseded by
$\{\mathbf v_{S-(i_1,j_1)}, \mathbf v_{S+(i_2,j_2)}, \mathbf
v_{S-(i_3,j_3)} \}$ if and only if
$$
    \frac{~j_2-j_3~}{i_3-i_2} \le \frac{~j_1-j_3~}{i_3-i_1}.
$$
\end{itemize}
\end{lemma}
\begin{proof}
We prove (i) only, (ii) is similar. Let $\b_1=\b(i_1,j_1)$,
$\b_2=\b(i_2,j_2)$ and $\b_3=\b(i_3,j_3)$. Then, according to Definition
\ref{def:triplets},
\begin{align*}
  \mathbf v_{S+(i_1,j_1)} &= \<\alpha_S+\b_1,\beta_S+i_1\b_1,\gamma_S+j_1\b_1\>, \\
  \mathbf v_{S-(i_2,j_2)} &= \<\alpha_S-\b_2,\beta_S-i_2\b_2,\gamma_S-j_2\b_2\>, \\
  \mathbf v_{S+(i_3,j_3)} &= \<\alpha_S+\b_3,\beta_S+i_3\b_3,\gamma_S+j_3\b_3\>.
\end{align*}
$\mathbf v_S=\<\alpha_S,\beta_S,\gamma_S\>$ is superseded by these vectors
if there are non-negative numbers $\lambda_1$,
$\lambda_2$, $\lambda_3$ with $\lambda_1+\lambda_2+\lambda_3=1$ such that
\begin{align*}
   \alpha_S &\le \lambda_1(\alpha_S+\b_1) + \lambda_2(\alpha_S-\b_2) + \lambda_3(\alpha_S+\b_3), \\
   \beta_S &\ge \lambda_1(\beta_S+i_1\b_1)+\lambda_2(\beta_S-i_2\b_2)+\lambda_3(\beta_S+i_3\b_3), \\
   \gamma_S &\ge \lambda_1(\gamma_S+j_1\b_1)+\lambda_2(\gamma_S-j_2\b_2)+\lambda_3(\gamma_S+j_3\b_3).
\end{align*}
Since the sum of the $\lambda_i$'s is $1$, this system is equivalent to
\begin{align*}
   \lambda_2\b_2 &\le \lambda_1\b_1+\lambda_3\b_3, \\
   i_2\lambda_2\b_2 &\ge i_1\lambda_1\b_1 + i_3\lambda_3\b_3, \\
   j_2\lambda_2\b_2 &\ge j_1\lambda_1\b_1 + j_3\lambda_3\b_3.
\end{align*}
Clearly, $\lambda_2$ must be strictly positive as $\b_1$, $\b_2$, and $\b_3$ are
positive. Introducing $\mu_1=(\lambda_1\b_1)/(\lambda_2\b_2)$ and
$\mu_3=(\lambda_3\b_3)/(\lambda_2\b_2)$, this system is equivalent to
\begin{align*}
    1 & \le \mu_1+\mu_3, \\
  i_2 & \ge i_1\mu_1 + i_3\mu_3 \\
  j_2 & \ge j_1\mu_1 + j_3\mu_3.
\end{align*}
One can assume that the first inequality holds with an equality.
Since $i_1<i_2<i_3$, the second inequality holds when $i_2$ is \emph{above} the
point which splits the interval $[i_1,i_3]$ in ratio $\mu_3$ to $\mu_1$.
Similarly, $-j_1<-j_2<-j_3$ implies that the third inequality holds when
$-j_2$ is \emph{below} the point that splits $[-j_1,-j_3]$ in the same
ratio. Thus, non-negative numbers $\mu_1$ and $\mu_3$ satisfying these three
inequalities exist if and only if the proportion of $[i_2,i_3]$ in
$[i_1,i_3]$ is not larger than the proportion of $[-j_2,-j_3]$ in the
interval $[-j_1,-j_3]$, that is,
$$
    \frac{(-j_3)-(-j_2)}{(-j_3)-(-j_1)} \ge \frac{i_3-i_2}{i_3-i_1}.
$$
This condition is equivalent to the one given in the claim.
\end{proof}

\begin{corollary}\label{corr:slope}
Let $i_1<i_2<i_3$, and $j_1>j_2>j_3$. Assume both $(i_k,j_k)$ and
$(i_k+1,j_k-1)$ are in $S$, while $(i_k,j_k+1)\notin S$ and
$(i_k+1,j_k)\notin S$ for $k=1,2,3$. When $j_3=0$ the condition with negative
values is assumed to hold. If $(j_1-j_3)/(i_3-i_1)$ is not in the open
interval
$$
    \bigg( \frac{j_2-j_3}{(i_3\m-i_2)+1} ,~
         \frac{j_2-j_3}{(i_3\m-i_2)-1} \bigg),
$$
then $\mathbf v_S$ is superseded by vectors generated by one of
the following two triplets:
$$\begin{array}{l@{~}l@{~}l}
\{S+(i_1\m+1,j_1), &S-(i_2,j_2), &S+(i_3\m+1,j_3)\}, \\[3pt]
\{S-(i_1,j_1), &S+(i_2\m+1,j_2), &S-(i_3,j_3)\}.
\end{array}$$
\end{corollary}

\begin{proof}
If the slope $(j_1-j_3)/(i_3-i_1)$ is less than, or equal to the lower limit,
then part (i) of Lemma \ref{lemma:convex} applies to the first triplet. When
the slope is at or above the upper limit, then part (ii) of that Lemma applies
to the second triplet.
\end{proof}

A downward closed set $S\subset\N\times\N$ can be specified in two ways.
Either by a non-increasing sequence $S^{\mathrm{col}} =(c_0,c_1,\dots,c_k)$
specifying the maximal values in columns $0$, \dots, $k$ (column sequence),
or by a non-increasing sequence $S^{\mathrm{row}}=(r_0,r_1,\dots,r_\ell)$
specifying the maximal values in rows $0$, \dots, $\ell$ (row sequence). It
is easy to see that
$$
    (x,y)\in S ~\Leftrightarrow~ 0\le y\le c_x ~\Leftrightarrow~ 0\le x\le r_y.
$$
For the downward closed subsets of $D_3$ the corresponding column and row
sequences are presented in Table \ref{table:colseq}.

\begin{table}[thb]
$$
\begin{tikzpicture}[scale=1.3]
\def\DDo#1{\draw[color=gray] (#1) circle(1.5pt);}
\def\DDi#1{\draw[fill=orange] (#1) circle(1.5pt); }
\def\,{\mkern1mu}
\def\pxxts#1#2#3#4#5#6#7#8#9{%
\csname DD#1\endcsname{-20pt+20*#8pt,0}%
\csname DD#2\endcsname{-20pt+20*#8pt,5pt}%
\csname DD#3\endcsname{-20pt+20*#8pt,10pt}%
\csname DD#4\endcsname{-20pt+20*#8pt+5pt,0}%
\csname DD#5\endcsname{-20pt+20*#8pt+5pt,5pt}%
\csname DD#6\endcsname{-20pt+20*#8pt+10pt,0}%
\draw (-20pt+20*#8pt-5pt,-10pt) node[right,color=black!60!green]{\small$#9$};
\draw (-20pt+20*#8pt-5pt,-19pt) node[right,color=purple!30!black]{\small$#7$};
\draw[color=black!30!blue] (-20pt+20*#8pt-5pt,21pt) node[right] {$\scriptstyle\##8$};
}
\foreach\x in {0,2,4,6,8,10,12}{
  \filldraw[color=black!5] (\x*20pt-4pt,-24pt) rectangle (\x*20pt+13pt,28pt);
}
\draw[line width=2pt,color=black!5](-4pt,-24pt)--(240pt,-24pt)
                                    (-4pt,28pt)--(240pt,28pt);
\pxxts iooooo{0}1{0}
\pxxts iooioo{1}2{0\,0}
\pxxts iooioi{2}3{0\,0\,0}
\pxxts iioooo{0\,0}4{1}
\pxxts iioioo{1\,0}5{1\,0}
\pxxts iioioi{2\,0}6{1\,0\,0}
\pxxts iioiio{1\,1}7{1\,1}
\pxxts iioiii{2\,1}8{1\,1\,0}
\pxxts iiiooo{0\,0\,0}9{2}
\pxxts iiiioo{1\,0\,0}{10}{2\,0}
\pxxts iiiioi{2\,0\,0}{11}{2\,0\,0}
\pxxts iiiiio{1\,1\,0}{12}{2\,1}
\pxxts iiiiii{2\,1\,0}{13}{2\,1\,0}
\draw (-3pt,-28pt)[color=black!60!green] node[right]{\small column sequences};
\draw (252pt,-28pt) node[left,color=purple!50!black]{\small row sequences};
\draw[color=black!40!green,->] (-2pt,-28pt) to[out=180,in=180] (-3pt,-10pt);
\draw[color=purple,->] (252pt,-28pt) to[out=0,in=0] (252pt,-18pt);

\end{tikzpicture}
$$
\caption{Column and row sequences for subsets of $D_3$.}%
\label{table:colseq}
\end{table}

\begin{corollary}\label{corr1}
If for some $S\subseteq D_n$ the vector $\mathbf v_S$ is not superseded by
other vectors generated from subsets of $D_n$, then either the sequence
$S^{\mathrm {col}}$ is strictly decreasing, or the sequence
$S^{\mathrm{row}}$ is strictly decreasing.
\end{corollary}

\begin{proof}
If $S^{\mathrm{col}}$ is not strictly decreasing, then the upper bound of
$S$ contains a horizontal segment of length at least $2$. Similarly, if
$S^{\mathrm{row}}$ is not strictly decreasing, then the right bound of $S$
contains a vertical segment of length at least $2$, see Figure
\ref{fig:corr1}. Take such a horizontal and a vertical segments whose
distance is minimal. Let the horizontal segment be in row $r$ between
columns $c_1$ and $c_2$, and the vertical segment be in column $c$ between
rows $r_1$ and $r_2$. The horizontal and vertical segments are connected by
(a possibly empty) diagonal staircase. Depending on which segment comes
first, there are two possible arrangements as depicted on Figure
\ref{fig:corr1}.
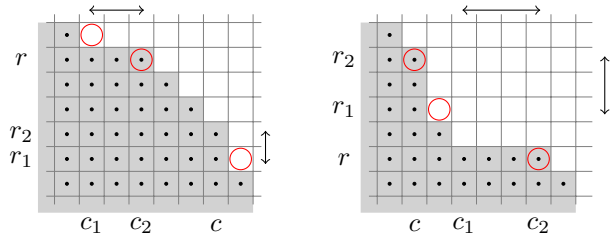
\begin{figure}[h!tb]
\begin{tikzpicture}[scale=1.1]
\def\dd(#1,#2){\filldraw[black] (\x+0.3*#1,\y+0.3*#2) circle (0.5pt);}
\def\cc(#1,#2){\draw[red] (\x+0.3*#1,\y+0.3*#2) circle (0.13cm);}
\def\ss#1{#1}
\draw[<->] (0.42,2.25) -- (0.18+0.9,2.25);
\draw[<->] (2.55,0.15+0.64) -- (2.55,0.15+0.24);
\fill[gray!35] (-0.2,-0.2)--(-0.2,2.1)--(0.3,2.1)--(0.3,1.8)--(1.2,1.8)
      --(1.2,1.5)--(1.5,1.5)--(1.5,1.2)--(1.8,1.2)--(1.8,0.9)--(2.1,0.9)
      --(2.1,0.3)
      --(2.4,0.3)--(2.4,-0.2)--cycle;
\draw (-0.4,1.65) node {$\ss r$};
\draw (-0.4,0.45) node {$\ss r_1$};
\draw (-0.4,0.75) node {$\ss r_2$};
\draw (0.45,-0.35) node{$\ss c_1$};
\draw (1.05,-0.35) node{$\ss c_2$};
\draw (1.95,-0.35) node{$\ss c$};
\draw[step=0.3cm,help lines] (-0.1,-0.1) grid (2.5,2.2);
\def\x{0.15}\def\y{0.15}
\foreach\j/\k in {0/6,1/5,2/5,3/5,4/4,5/3}{
\foreach\i in {0,...,\k}{\dd(\j,\i)}}
\foreach\i in {0,1,2}{\dd(6,\i)} \dd(7,0)
           \cc(1,6)\cc(3,5)\cc(7,1)
\draw[<->] (4.03+0.9,2.25)--(4.07+1.8,2.25);
\draw[<->] (6.65,0.15+1.54)--(6.65,0.15+0.84);
\fill[gray!35] (3.7,-0.2)--(3.7,2.1)--(4.2,2.1)--(4.2,1.8)--(4.5,1.8)--(4.5,0.9)
      --(4.8,0.9)--(4.8,0.6)--(6.0,0.6)--(6.0,0.3)--(6.3,0.3)--(6.3,-0.2)
      --cycle;
\draw[step=0.3cm,help lines] (3.8,-0.1) grid (6.5,2.2);
\def\x{4.05}
\foreach\i in{0,...,6}{\dd(0,\i)}
\foreach\i in{0,...,5}{\dd(1,\i)}
\foreach\i in{0,1,2}{\dd(2,\i)}
\foreach\i in{0,1}{\dd(3,\i)\dd(4,\i)\dd(5,\i)\dd(6,\i)} \dd(7,0)
     \cc(1,5)\cc(2,3)\cc(6,1)
\draw (3.5,0.45) node{$\ss r$};
\draw (3.5,1.05) node{$\ss r_1$};
\draw (3.5,1.65) node{$\ss r_2$};
\draw (4.35,-0.35) node{$\ss c$};
\draw (4.95,-0.35) node{$\ss c_1$};
\draw (5.85,-0.35) node{$\ss c_2$};
\end{tikzpicture}
\caption{Nearest horizontal and vertical boundary segments of the gray
region are marked by the arrows. Red circles indicate points to be added and
subtracted.
}\label{fig:corr1}
\end{figure}
In the first case $c_1<c_2<c$, and $r_1<r_2<r$; in the second case
$c<c_1<c_2$ and $r<r_1<r_2$. Apply Lemma \ref{lemma:convex} to the marked
points and observe that the modified downward closed sets are always subsets
of $D_n$. In the first case
$$
   \frac{c_2-c_1}{(r+1)-r} \ge 1 \ge \frac{c+1-c_2}{r_1-r},
$$
and in the second case
$$
   \frac{(c+1)-c)}{r_1-r_2} \le 1 \le \frac{c_2-(c+1)}{r-r_1}.
$$
Therefore, by Lemma \ref{lemma:convex}, $\mathbf v_S$ is superseded by the
vectors generated by the indicated sets, proving the claim.
\end{proof}

By Corollary \ref{corr1}, the downward closed set corresponding to an
extremal vertex is either a staircase with step heights $1$ (when
$S^{\mathrm{row}}$ is strictly decreasing), which we call \emph{horizontal,}
or the mirror image of such a staircase. The only configuration that belongs
to both cases is the diagonal $D_n$. It will be more convenient to use the
column-sequence $(c_0,c_1,\dots,c_k)$ to represent horizontal staircases.
Here $k\ge 0$ is the \emph{length} of the staircase, also denoted by
$\len(S)$. The last column size (height) is necessarily $c_k=0$, and $c_i$
equals either $c_{i+1}$ or $c_{i+1}\m+1$ for every $0\le i<k$. In the rest
of this section, all staircases, if not mentioned otherwise, are horizontal
ones.

\begin{definition}\label{def:PNP}
The staircase $S$ is Positive-Negative-Positive (PNP)-reducible in $D_n$ if
there are $i_1<i_2<i_3$ and $j_1>j_2>j_3$ such that $S_1=S+(i_1,j_1)$,
$S_2=S-(i_2,j_2)$, and $S_3=S+(i_3,j_3)$ are staircases in $D_n$ and
$\mathbf v_S$ is superseded by $\{\mathbf v_{S_1},\mathbf v_{S_2},\mathbf
v_{S_3}\}$. $S$ is \emph{PNP-irreducible} if it is not PNP-reducible.

Negative-Positive-Negative (NPN)-reducibility and NPN-irreducibility is
defined analogously, using staircases $S-(i_1,j_1)$, $S+(i_2,j_2)$, and
$S-(i_3,j_3)$, assuming that they are also subsets of $D_n$. $S$ is
\emph{irreducible in $D_n$} if it is both PNP- and NPN-irreducible. Finally,
let $\mathcal S_n$ be the collection of the irreducible staircases that are
subsets of $D_n$.
\end{definition}

By the remark at the beginning of this section, by Lemma \ref{lemma:convex},
and by Corollary \ref{corr1}, extremal vertices are generated by elements of
$\mathcal S_n$ and by their mirror images. We describe an incremental
algorithm that generates the elements of the collection $\mathcal S_n$.

A horizontal staircase $S$ of length $n$ can be recovered from a unique
horizontal staircase $S'$ of length $n-1$ as follows. If $S'$ has the column
sequence $(c'_0,c'_1,\dots,c'_{n-1})$, then $S$ is defined by one of the
column sequences
$$\begin{array}{l@{~}l@{~}l@{~}l@{~}l}
  (c'_0, & c'_1, & \dots, & c'_{n-1}, & 0) \mbox{ or}\\[2pt]
  (c'_0\m+1, & c'_1\m+1, & \dots, & c'_{n-1}\m+1, &0),
\end{array}$$
depending on whether the last two elements of the column sequence of $S$ are
equal.

\begin{claim}\label{claim:cases}{\upshape(i)} \hangindent=2.5\parindent\hangafter=1
Suppose $S$ has length $n$. $S$ is irreducible in $D_{n+1}$ if and only if
it is irreducible in $D_m$ for any $m\ge n+1$.\par

\titem{(ii)} If $\len(S)=n$ and $S$ is irreducible in $D_n$, then $S'$ is
irreducible in $D_{n-1}$.

\titem{(iii)} If $S\in\mathcal S_n$ but $S\notin\mathcal S_{n+1}$, then
$\len(S)=n$ and $S$ is PNP-reducible in $D_{n+1}$ with $i_3=n+1$ and
$j_3=0$.

\titem{(iv)} If $S'\in\mathcal S_{n-1}$ and $S\notin\mathcal S_n$, then
either $S$ is NPN-reducible with $i_3=n$ and $j_3=0$, or it is
PNP-reducible with $i_3=n-1$ and $j_3=1$.
\end{claim}
\begin{proof}
(i) is immediate from the definition as the staircases $S\pm(i,j)$ must be
subsets of $D_{n+1}$. 

\smallskip
\noindent(ii)
Assume $S'$ is reducible in $D_{n-1}$ shown by
the staircases $S'_1$, $S'_2$ and $S'_3$. Since they are in $D_{n-1}$, they
can be lifted back to $S_1$, $S_2$, $S_3$ in $D_n$. According to Lemma
\ref{lemma:convex} these staircases witness the reducibility of $S$.

\smallskip
\noindent (iii)
If $S$ is reducible in $D_{n+1}$ but not in $D_n$, then $S+(i_3,j_3)$ is not
in $D_n$, leading to the stated condition.

\smallskip
\noindent (iv)
If $S$ is reducible in $D_n$ while $S'$ is not reducible in $D_{n-1}$, then
the reduction must use $(i_3,j_3)$, which is in $D_n$ but not in $D_{n-1}$.
If it is an NPN-reduction then it must use the newly added point $(n,0)$; in
other cases the reduction can be shifted back to $S'$. In the case of a
PNP-reduction this additional point is $(n-1,1)$ (when extending the
staircase by a column of height zero), or can be shifted back to $S'$ again.
\end{proof}

\begin{pseudocode}{Generating irreducible staircases\label{code:1}}
\State /\!/ compute the set of irreducible staircases incrementally\label{c1-l1}
\State set $\mathcal S_0=\{(0)\}$, where $(0)$ is the staircase with a
single point at the origin.
\If{$\mathcal S_n$ has been created for $n\ge 0$, }
  \ForEach{ $S\in \mathcal S_n$ }
     \If{$\len(S)<n$}
       \State copy $S$ to $\mathcal S_{n+1}$
     \Else  ~~~~/\!/ $\len(S)=n$
       \If {$S$ is \emph{not} PNP-reducible with $(n+1,0)$,}
          \State copy $S$ to $\mathcal S_{n+1}$.
       \EndIf
       \State write $S$ as the column sequence $(c_0,\dots,c_n)$.
       \State add the staircase $(c_0\m+1,c_1\m+1,\dots,c_n\m+1,0)$ to $\mathcal S_{n+1}$.
       \State form the new staircase $S^*=(c_0,\dots,c_n,0)$.
       \If{$S^*$ is \emph{not} NPN-reducible with $(n+1,0)$, \emph{and} \\
           \hbox{\hspace{7.5pt}} $S^*$ is \emph{not} PNP-reducible with $(n,1)$,}
          \State add $S^*$ to $\mathcal S_{n+1}$.
       \EndIf
     \EndIf
  \EndFor
\EndIf
\end{pseudocode}

Based on Claim \ref{claim:cases}, the incremental algorithm, sketched as
Algorithm \ref{code:1}, generates all horizontal irreducible staircases. The
PNP- and NPN-irreducibility can be checked based on Lemma
\ref{lemma:convex}. The last point $(i_3,j_3)$ is fixed, and the na\"ive
implementation requires quadratic running time in $\len(S)$. With some
simple bookkeeping, it can be reduced to a backward scanning of the column
sequence, resulting in linear running time.

\begin{table}[thb]
\makeatletter
\def\ybb#1{%
\setbox\@tempboxa\hbox{#1}%
\begin{tikzpicture}\useasboundingbox(0,0)rectangle(0,0);
\fill[color=yellow!30] (-3pt,-3pt) rectangle (3pt+\wd\@tempboxa,10pt);
\end{tikzpicture}\unhbox\@tempboxa}
\makeatother\setlength\tabcolsep{10pt}
\begin{tabular}{llllll}
\ybb{0} & 2210 & 33210 & 332100 & 1111000 & 4332110\\
00 & \ybb{3210} & \ybb{43210} & 332110 & 2111000 & 4432100\\
\ybb{10} & 00000 & 000000 & 432100 & 2211000 & 4432110\\
000 & 10000 & 100000 & 432110 & 2211100 & 4432210\\
100 & 11000 & 110000 & 432210 & 2221100 & 5432100\\
110 & 11100 & 111000 & 433210 & 3221100 & 5432110\\
\ybb{210} & 21100 & 211000 & 443210 & 3321100 & 5432210\\
0000 & 22100 & 211100 &\ybb{543210} & 3322100 & 5433210\\
1000 & 22110 & 221100 & 0000000 & 3322110 & 5443210\\
1100 & 32100 & 321100 & 1000000 & 4322100 & 5543210\\
2100 & 32110 & 322100 & 1100000 & 4322110 & \ybb{6543210}\\
2110 & 32210 & 322110 & 1110000 & 4332100 &
\end{tabular}

\medskip

\caption{Column sequences of extremal horizontal
staircases.}\label{table:extremal}
\end{table}

Using the algorithm, we have computed the complete set of irreducible
staircases up to $n=60$. Table \ref{table:extremal} lists, up to length $6$,
the column sequences of those irreducible horizontal staircases that
remained irreducible in each subsequent generation. The irreducible downward
closed sets are these and their mirror images. Column sequences marked by a
yellow background are the diagonals $D_n$. An example is the column sequence
$2221100$ that defines the following downward closed set:
$$
   \begin{tikzpicture}[scale=1.5]
\fill[gray!35] (0,0)--(0,0.9)--(0.9,0.9)--(0.9,0.6)--(1.5,0.6)--(1.5,0.3)
   -- (2.1,0.3)--(2.1,0)--cycle;
\draw[step=0.3cm,help lines] (0,0) grid (2.2,1.0);
\foreach\x in{0,1,2}{
  \draw (-0.14,0.15+0.3*\x) node {$\scriptstyle\x$};
}
\foreach\x in{0,1,2,3,4,5,6}{
  \draw (0.15+0.3*\x,-0.15) node {$\scriptstyle\x$};
}
\foreach \x in{0,1,2,3,4,5,6}{
  \draw (0.15+0.3*\x,0.15) node {$\scriptstyle1$};
}
\foreach \x in{1,2,3,4,5}{
  \draw (-0.15+0.3*\x,0.45) node {$\scriptstyle\x$};
}
\draw (0.15,0.75) node{$\scriptstyle1$};
\draw (0.45,0.75) node{$\scriptstyle3$};
\draw (0.75,0.75) node{$\scriptstyle6$};
\draw (1.05,0.75)[color=gray!45] node{$\scriptstyle10$};
\draw (1.35,0.75)[color=gray!45] node{$\scriptstyle15$};
\draw (1.65,0.75)[color=gray!45] node{$\scriptstyle21$};
\draw (1.65,0.45)[color=gray!45] node{$\scriptstyle6$};
\draw (1.95,0.75)[color=gray!45] node{$\scriptstyle28$};
\draw (1.95,0.45)[color=gray!45] node{$\scriptstyle7$};
\end{tikzpicture}
$$
while the corresponding triplet is $\<32,76,35\>$. For each length, the number
of sequences in Table \ref{table:extremal} is
$$
   1,~2,~4,~7,~12,~18,~27,~38,~52,~68,~\dots
$$
The number of new extremal triplets is $2x-1$, that is,
$$
   1,~3,~7,~13,~23,~35,~53,~75,~103,~135,~\dots
$$
This latter sequence matches, up to the $60$th term, the sequence A103116 in
the Encyclopedia of Integer Sequences \cite{oeis}, which is defined as
$$
  \mathrm{remains}_n = \sum_{i\le n} (n-i+1) \phi(i),
$$
where $\phi$ is Euler's totient function. The formula suggests that the
connection is based on the number of different slopes determined by the
lattice points in a rectangle. Proving the equivalence of these two
sequences is an intriguing open problem.

\medskip

For better visualization, triplets $\<\alpha_S,\beta_S,\gamma_S\>$
corresponding to these irreducible staircases are plotted as the
three-dimensional points $\<\beta/\alpha, \gamma/\alpha,\alpha\>$ using
logarithmic scale for the third $\alpha$ coordinate. The plot in Figure
\ref{fig:ext1} contains all 126,981 extremal triplets in the range $\beta,
\gamma\le 20\alpha$. Some of the plotted triplets appear as late as
generation $n=80$; later generations do not contribute to this part of the
complete set. For comparison, some triplets in the $80$th generation have
values larger than $2^{85}$.

\begin{figure}[h!tb]
    \centering
    \begin{tikzpicture}[scale=1.2]
\draw(0,0) node[anchor=south west] {\includegraphics[width=9.6cm]{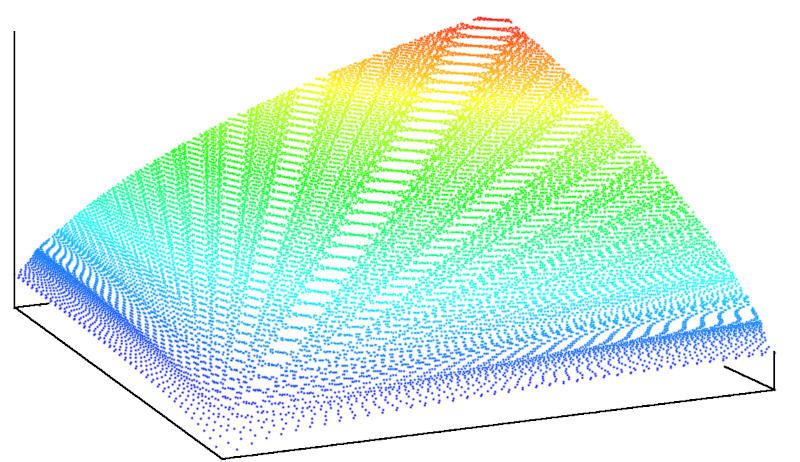}};
\draw[->,thick] (0.256,1.665) -- (0.252,4.54);
\draw (0.5,4.4) node{$\scriptstyle\alpha$};
\foreach\x in{3,6,9,12}{
  \draw(0.18,1.665+\x*0.20329)--(0.3,1.665+\x*0.20329);
  \draw(-0.1,1.665+\x*0.20329) node{$\scriptstyle 10^{\x}$};
}
\draw[<-,thick] (0.27-0.125*2.086,1.645+0.125*1.505)--(2.334+0.05*2.086,0.135-0.05*1.505);
\draw[->,thick] (2.34-0.5*0.27545,0.14-0.5*0.0344) -- (7.865+1.5*0.27545,0.834+1.5*0.0344);
\draw (7.92+1.5*0.27545,0.848+1.5*0.0344) node[above] {$\scriptstyle\beta/\alpha$};
\draw (0.04,1.84) node[left] {$\scriptstyle\gamma/\alpha$};
\foreach \x in{0,5,10,15,20}{
  \draw (2.45+\x*0.27545,-0.03+\x*0.0344) node {$\scriptstyle\x$};
}
\foreach \x in{0,5,10,15,20}{
  \draw (2.12-\x*0.1043,0.1+\x*0.07526) node {$\scriptstyle\x$};
}
    \end{tikzpicture}%
    \caption{Extremal configurations. Colors indicate the $\alpha$ value.}\label{fig:ext1}%
\end{figure}
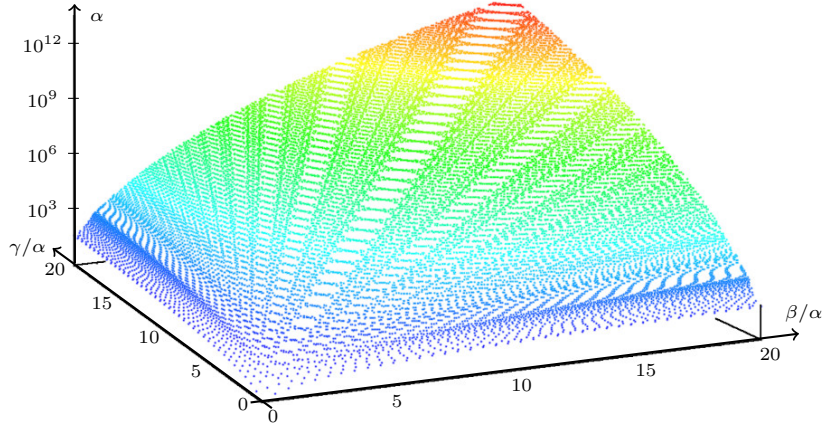

To explain the shape of the surface of extremal triplets plotted in Figure
\ref{fig:ext1}, we provide some heuristic reasoning. A consequence of
Corollary \ref{corr:slope} is that if the extremal triplet $\mathbf v_S$ is
computed from the staircase $S$, then the slopes determined by the step
edges $(i,j)\in S$ (namely, points of $S$ where neither $(i+1,j)$ nor
$(i,j+1)$ are in $S$) are almost equal. Consequently, on a large scale,
extremal $\mathbf v_S$ vectors are generated by the set of lattice points in
right-angled triangles defined by the inequality
$$
   S(a,b) = \{ (x,y)\in \N\times\N : \frac x a + \frac y b \le 1 \}
$$
for some positive values of $a$ and $b$. Since, by Lemma \ref{lemma:bsum},
\begin{align*}
  \sum_{i\le x,\;j\le y} \b(i,j) & \approx \b(x\m+1,y\m+1), \mbox{ and}\\
  \sum_{i\le x,\;j\le y} i\tsp\b(i,j) & \approx x\tsp\b(x\m+1,y\m+1),
\end{align*}
the vector $\mathbf v_{S(a,b)}$ is well approximated by
$\b(x\m+1,y\m+1)\<1,x, y\>$, where $(x,y)\in S(a,b)$ is the point where
$\b(x\m+1,y\m+1)$ takes its maximal value. As the function $\b(x,y)$
strictly increases in both coordinates, this maximum is taken on the
boundary diagonal of the right-angled triangle $S(a,b)$ that has endpoints
$(0,b)$ and $(a,0)$. Using the Stirling formula $n! \approx \sqrt{2\pi
n}\tsp(n/e)^n$, we have
$$
   \b(x,y) = \frac{\tsp(x+y)!\tsp}{x!\, y!}
   \approx \frac{1}{\sqrt{2\pi}}\frac{(x+y)^{x+y+1/2}}
       {x^{x+1/2}\tsp y^{y+1/2}}.
$$
Introducing $\phi(x)=(x+\frac12)\log x$, we see that the logarithm of
$\b(x,y)$ is well approximated by the function
$$
  \theta(x,y) = \phi(x+y)-\phi(x)-\phi(y) -\log \sqrt{2\pi}.
$$
Using this approximation, the point $(u,v)$ is extremal in the triangle
$S(a,b)$ if $(u,v)$ is on the boundary diagonal and $\theta(u\m+1,v\m+1)$
has zero derivative along this diagonal. For fixed $u$ and $v$ such a
positive $a$ and $b$ exist just in case the partial derivatives $\theta_x$
and $\theta_y$ are positive at $(u\m+1,v\m+1)$. By inspection, this
condition is satisfied for every $(u,v)$. Consequently, if
$\<\alpha,\beta,\gamma\>$ is an extremal triplet, then choosing
$u=\beta/\alpha$, $v=\gamma/\alpha$, we expect
$$
   \log\alpha\approx\theta(u+1,v+1),
$$
and, conversely, for each $u$, $v$, with the choice $\log\alpha=
\theta(u\m+1,v\m+1)$, $\beta=u\tsp\alpha$, and $\gamma = v\tsp\alpha$,
we expect the triplet $\<\alpha,\beta,\gamma\>$ to be extremal.
For comparison,
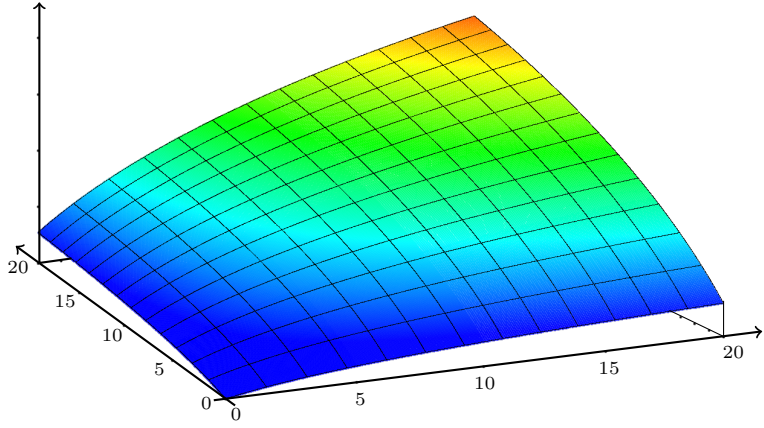
\begin{figure}
    \centering
\begin{tikzpicture}[scale=1.2]
\draw(0.04,-0.06)  node[anchor=south west] 
{\includegraphics[width=9.4cm]{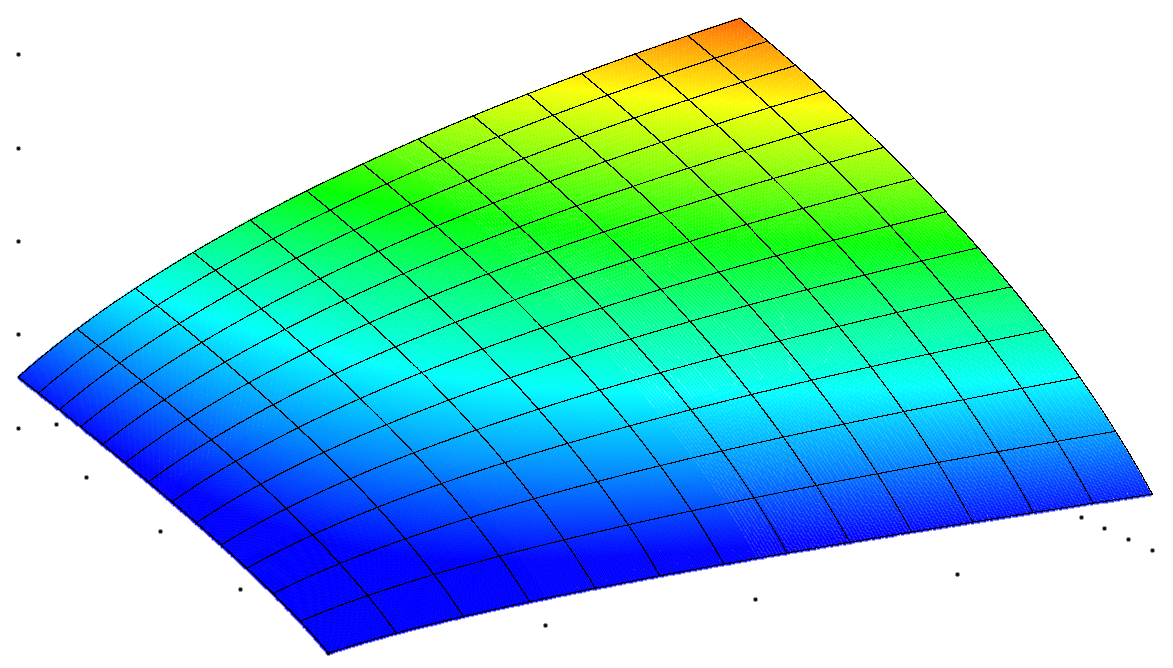}};
\draw[->,thick] (0.627,1.7)--(0.27,1.65) -- (0.27,4.54);
\draw[<-,thick] (0.27-0.125*2.086,1.65+0.125*1.505)--(2.334+0.05*2.086,0.145-0.05*1.505);
\draw[->,thick] (2.34-0.5*0.27545,0.145-0.5*0.0344) -- (7.865+1.5*0.27545,0.84+1.5*0.0344);
\draw (7.25,1.124)--(7.85,0.835) -- (7.85,1.23);
\foreach \x in{0,5,10,15,20}{
  \draw (2.45+\x*0.27545,-0.03+\x*0.0344) node {$\scriptstyle\x$};
}
\foreach \x in{0,5,10,15,20}{
  \draw (2.12-\x*0.1043,0.1+\x*0.07526) node {$\scriptstyle\x$};
}

\end{tikzpicture}%
    \caption{The $\theta(u+1,v+1)$ function. Colors indicate the function value.}
    \label{fig:thetafg}
\end{figure}
Figure \ref{fig:thetafg} plots these triplets over the same range that
was used in Figure \ref{fig:ext1}. This approximation seems to
slightly underestimate the real value of $\log\alpha$. For example,
the extremal triplet obtained from the diagonal staircase $D_{n+1}$ is
$$
  \alpha = 2^{n+1}-1, ~~~
  \beta = (n-1)\tsp2^n+1, ~~~ \gamma = (n-1)\tsp2^n+1,
$$
thus, $u=v=\beta/\alpha\approx (n-1)/2$, and $\log\alpha
\approx (n+1)\log 2$. At the same time,
$$
\theta(u\m+1,v\m+1)=(n+1)\log 2 - \log\sqrt{n\m+1}+O(1).
$$
Extremal triplets on the two edges of the surface are specified by the
totally flat, stairless staircases. These triplets are
$$
   \alpha=n, ~~~ \beta=n(n-1)/2, ~~~ \gamma=0
$$
on one axis, and $\beta$, $\gamma$ swapped on the other.
In this case, the $(u,v)$ pair is $((n\m-1)/2,0)$, and
$$
   \theta(u+1,v+1) = \log n +\big( 1 - \log\sqrt{8\pi} + O(1/n)\big),
$$
which differs from the correct value by a constant only.

We have also looked at how the newly discovered entropy inequalities delimit
the five-variable entropy region. The triplet $\<\alpha_S, \beta_S,
\gamma_S\>$ yields the inequality
$$
    (a,b\|z)+ \alpha_S\big( \abcd+\Z \big) + \beta_S\tsp\C +
     \gamma_S\tsp\D \ge 0.
$$
Since the closure of the $5$-variable entropy region is a pointed convex
cone, one can normalize it by assuming $(a,b\|z)=1$. An equivalent view is
to take the cross-section of $\clGa5$ by this hyperplane. Consider the
three-dimensional subspace spanned by the vectors
$$\begin{array}{r@{\;}l@{\;}l}
  \mathbf x &\eqdef \C &= (a,c\|b)+(b,c\|a), \\[2pt]
  \mathbf y &\eqdef \D &= (a,d\|b)+(b,d\|a), \\[2pt]
 -\mathbf z &\eqdef \abcd+\Z &= \abcd + (a,z\|b) + (b,z\|a);
\end{array}$$
observe that $\mathbf z$ is negated. Normalize the five-variable entropic
function $f$ so that it satisfies $f(a,b\|z)=1$, then project it to this
subspace. Use the scalar products $\<f\m\cdot\mathbf x,f\m\cdot\mathbf
y,f\m\cdot\mathbf z\>$ as the projection coordinates. This three-dimensional
cross-section of the five-variable entropy region is
$$
  \Delta = \big\{ \<f\m\cdot\mathbf x,
                f\m\cdot\mathbf y,
                f\m\cdot\mathbf z\>\in\R^3 : ~~
                f\in\clGa5 \mbox{ such that } f(a,b\|z)=1 \big\}.
$$
Clearly, points in $\Delta$ have non-negative $x$ and $y$ coordinates,
while the $z$ coordinate can take both positive and negative values.
Since $\clGa5$ is a closed convex cone, $\Delta$ is closed and convex.
We concentrate on the part above the $xy$ plane:
$$
  \Delta^+ = \{\<x,y,z\>\in\Delta: z\ge 0\}.
$$
Shannon inequalities provide no restriction whatsoever on $\Delta^+$ as any
non-negative coordinate triplet can be realized by some polymatroid. To show
this, define $\r_I$ for any subset $I$ of the ground set $abcdz$ as
$$
   \r_I : A \mapsto \begin{cases}
          1& \mbox{ if } I\cap A\ne\emptyset,\\
          0& \mbox{ otherwise.}
          \end{cases}
$$
Note that $\lambda\tsp\r_I$ is entropic for any non-negative $\lambda$.
Let, moreover, $\r^\circ$ be the function
$$
   \r^\circ : A\mapsto  \begin{cases}
     2 & \mbox{ if } |A|=1, \\
     4 & \mbox{ if $|A|\ge3$, or $A$ is one of }cd, cz, dz,\\
     3 & \mbox{ otherwise,}
\end{cases}$$
as $A$ runs over the non-empty subsets of $abcdz$. Both $\r_I$ and
$\r^\circ$ are exremal rays of $\G5$, so they are polymatroids. For
arbitrary non-negative numbers $x,y,z$ the linear combination
$f=\r_{abcd}+x\tsp\r_{ac}+y\tsp\r_{ad}+z\tsp\r^\circ$ satisfies
$f(a,b\|z)=1$ and has coordinates $\langle x,y,z \rangle$, providing the
required polymatroid.

Points of $\Delta^+$ with non-negative $x$ and $y$ coordinates and $z=0$ are
realized by linear polymatroids; thus, the complete non-negative quadrant of
the $xy$ plane is a part of $\Delta^+$. Our first non-Shannon inequality,
generated by the triplet $S_{D_1}= \<1,0,0\>$, is
$$
    (a,b\|z)+\abcd +\Z \ge 0.
$$
This inequality immediately limits the region $\Delta^+$ to $z\le 1$;
therefore, points in $\Delta^+$ have a height at most $1$.
%
\begin{figure}[htb]
    \centering
    \begin{tikzpicture}[>=stealth]
\draw(0,0) node[anchor=south west] {\includegraphics[width=8cm]{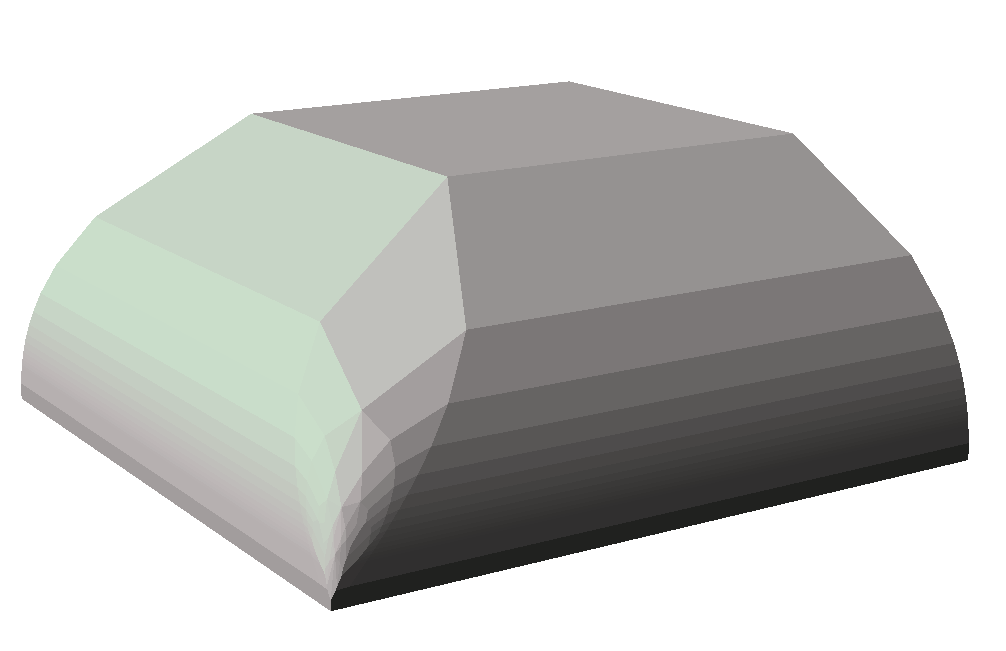}};
\draw (-0.3,2.5) node {$y$};
\draw[<-](-0.1,2.4)--(2.76+2.86*0.05,0.44-0.05*1.96);
\draw[->](2.76,0.44)--(2.76,5.15);
\draw[->](2.76-0.05*5.44,0.44-0.05*1.3)--(8.2,1.74);
\draw (2.98,5.05) node {$z$};
\draw (8.4,1.77) node {$x$};
\foreach\x/\y in {2.76/3.15} {
  \draw (\x-0.1,\y) --(\x+0.1,\y);
  \draw (\x-0.2,\y) node {$\scriptstyle1$};
}
\foreach\x in {0,1,2,2.5}{
  \draw (2.76-\x*0.99,0.44+\x*0.68) +(-0.4,-0.1) node{$\scriptstyle\x$};
}
\foreach\x in {0,1,2,2.5}{
  \draw (2.76+\x*2.03,0.44+\x*0.484) +(0.2,-0.2) node{$\scriptstyle\x$};
}
\draw (4.5,5.0) node {$\scriptstyle(1,1,1)$};
\draw[->] (4.5,4.82) -- (3.71,3.91);
     \end{tikzpicture}%
    \caption{Points of $\Delta^+$ are on or below the indicated surface.}%
    \label{fig:region}%
\end{figure}
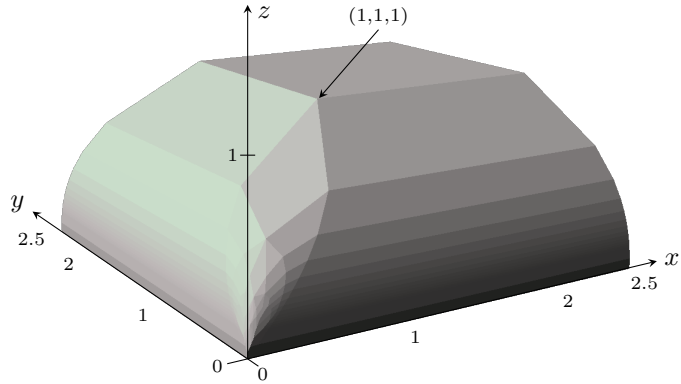
%
Other extremal triplets provide additional linear constraints. Figure
\ref{fig:region} illustrates the delimited part of the non-negative octant
as viewed from the origin, cut at $x\le2.5$ and at $y\le2.5$. The depicted
bound of $\Delta^+$ extends similarly to larger values of $x$ and $y$. Along
the $x$ and $y$ axes, the bound approaches the $xz$ and $yz$ coordinate
planes as the functions $z=\sqrt y$ and $z=\sqrt x$, respectively. Along the
$xy$ diagonal, the limiting behavior toward the $z$ axis is similar to the
entropy function $-(x+y)\log(x+y)$. The corner point of the $z=1$ plateau
has coordinates $\<1,1,1\>$. The $\theta(u\m+1,v\m+1)$ estimate provides a
smooth version of this surface. Figure \ref{fig:logregion} plots the
$0<x,y\le 1$ part of this bound on a logarithmic scale. The flat ``wings''
around the $x$ and $y$ edges, rising at a slope of $1/2$, represent the
region where the bound approaches the axes at the rate of $\sqrt x$ or
$\sqrt y$, respectively. The diagonal hump of the surface rises as $t-\log
t$, corresponding to the approximation $z\approx -(x+y)\log(x+y)$ as $x=y\to
0$.

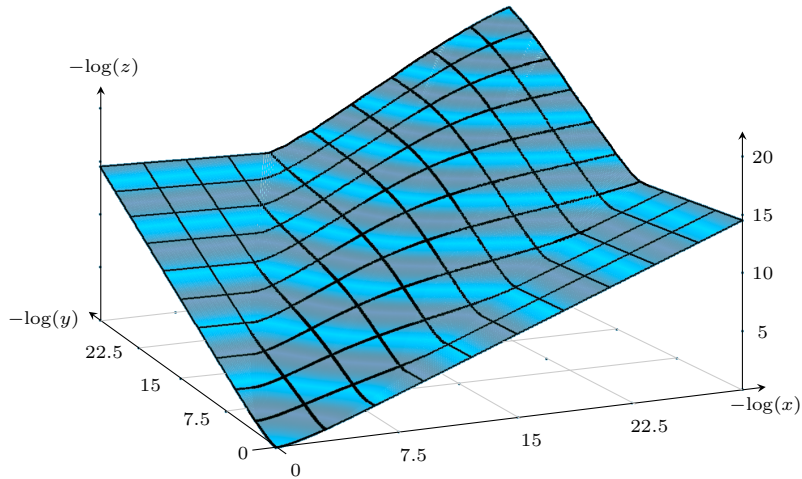
\begin{figure}[b!ht]
\hbox to 0.99\linewidth{\hfill
\includegraphics[width=9cm]{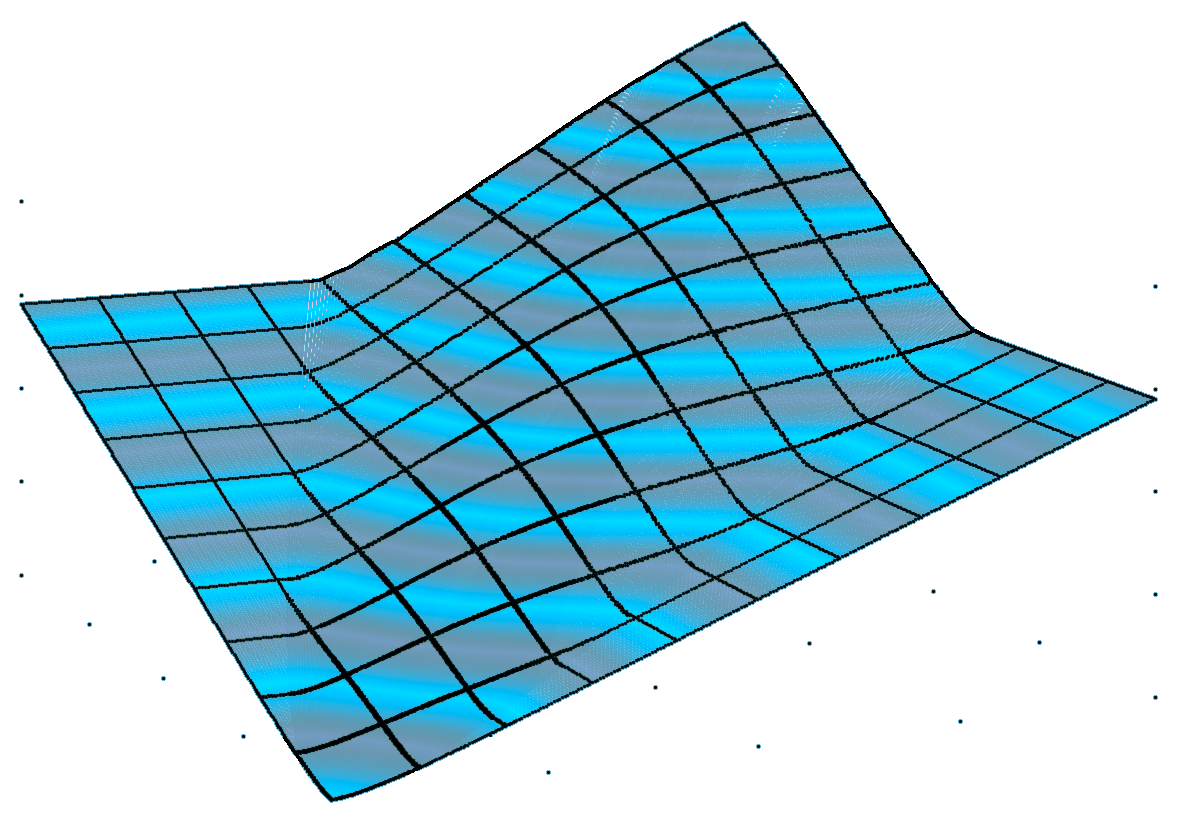}%
\begin{tikzpicture}[>=stealth]\useasboundingbox(0,0) rectangle(0,0);
\draw[->](-6.5-6.5*0.05,0.2-0.8*0.05)--(0,1);
\draw[->](-6.5+2.5*0.05,0.2-1.8*0.05)--(-9,2);
\draw (0,1) node[below] {$\scriptstyle {}-\log(x)$};
\draw (-9.6,2.15) node[below]{$\scriptstyle {}-\log(y)$};
\draw (-8.8,5) node[above]{$\scriptstyle{}-\log(z)$};
\draw[->](-8.84,1.88)--(-8.84,5);
\draw[->](-0.31,0.96)--(-0.31,4.4);
\draw[color=black!20](-8.84,1.88)--(-7.7,2.01)
   (-0.31,0.96)--(-2.55,2.05)
   (-1.73,0.77)--(-3.39,1.63)
   (-3.26,0.58)--(-4.295,1.17)
   (-4.85,0.39)--(-5.31,0.685)
   (-1.19,1.386)--(-4.7,0.98)
   (-2.0,1.77) -- (-3.5,1.6)
   (-7.76,1.12)-- (-7.18,1.18)
   (-7.17,0.68) -- (-6.90,0.716)
   (-8.33,1.52) --(-7.45,1.62)
;
\foreach\x in {0,7.5,15,22.5}{
  \draw(-6.95-\x*0.082,0.12+\x*0.06) node{$\scriptstyle\x$};
  \draw(-6.25+\x*0.209,-0.1+\x*0.026) node{$\scriptstyle\x$};
}
\foreach\x in{5,10,15,20}{
  \draw (-0.04,0.96+\x*0.155) node {$\scriptstyle\x$};
}
\end{tikzpicture}\hfill}%
\caption{Bounds of $\Delta^+$ around the origin on logarithmic scale}\label{fig:logregion}
\end{figure}

\section{Mat\'u\v s' Inequalities}\label{sec:fminf}

\newcommand\slb{\begin{tikzpicture}[scale=0.13\ht\strutbox]%
        \draw[line width=0.015cm](0,0)--(0.055,0.22)--(0.22,0.22);
        \draw[line width=0.035cm](0.22,0.22)--(0.165,0);
        \draw[line width=0.025cm](0.165,0)--(0,0);
                \end{tikzpicture}}

In \cite{M.infinf}, F.~Mat\'u\v s constructed an infinite family of
five-variable non-Shannon inequalities. Similar to our bound of $\Delta^+$
in Figure \ref{fig:region}, these inequalities provide the same quadratic
outer bound on a special two-dimensional cross-section of the four-variable
entropy region. By exhibiting examples for a quadratic inner bound in that
cross-section, Mat\'u\v s proved that the four-variable entropy region (and,
consequently, the entropy region $\Ga N$ with $|N|\ge 4$) is not polyhedral.
To describe his family of inequalities, let $\slb$ denote either $\abcd$ or
$\acbd$.

\begin{theorem}[F.~Mat\'u\v s \cite{M.infinf}]\label{thm:fminf}
The following is an entropy inequality for every non-negative $k$:
\begin{equation}\label{eq:m}
 (b,z\|a)+ k\tsp\big(\slb+(a,z\|b)+(a,b\|z)\big)
        + \frac{k(k\m-1)}2\big((a,c\|b)+(a,b\|c)\big) \ge 0.
\end{equation}
\end{theorem}

\begin{proof}

By induction on $k$. It clearly holds for $k=0$. To prove the inequality for
$k\m+1$, suppose it holds for $k$ in every entropic polymatroid. Since the
inequality to be proved does not contain entropy terms that intersect both
$cd$ and $z$, the Maximum Entropy Principle implies that this inequality
holds in every entropic polymatroid if and only if it holds in entropic
polymatroids that additionally satisfy $(cd,z\|ab)=0$; see also the proof of
Lemma \ref{lemma:Jalphabeta}. Thus, without loss of generality, we may
additionally assume that $cd$ and $z$ are conditionally independent given
$ab$. By induction, (\ref{eq:m}) holds with $b$ and $c$ exchanged (as it
holds in every entropic polymatroid); thus we have the following inequality:
\begin{equation}\label{eq:m2}
  (c,z\|a)+k\tsp\big(\slb+(a,z\|c)+(a,c\|z)\big)
     +\frac{k(k\m-1)}2\big((a,b\|c)+(a,c\|b)\big) \ge 0.
\end{equation}
The following inequalities hold in every polymatroid:
\begin{align*}
(b,z\|a)+ \big(\slb+(a,z\|b)+(a,b\|z)\big) &\ge (c,z\|a)-3(cd,z\|ab),\\
    (a,z\|b)+(a,b\|c) &\ge (a,z\|c)-(cd,z\|ab), \\
    (a,b\|z)+(a,c\|b) &\ge (a,c\|z)-(cd,z\|ab).
\end{align*}
Adding the first line once, the second line, and the third line $k$ times to
(\ref{eq:m2}), and observing that $(cd,z\|ab)=0$, yields the required inequality 
(\ref{eq:m}) for $k\m+1$.
\end{proof}

While Mat\'u\v s' inequalities do not appear in our families, they contain
strikingly similar ones. Combining results from Theorem \ref{thm:mainI} and
Theorem \ref{thm:mainII}, for each non-negative $k$ the following is an
entropy inequality:
\begin{equation}\label{eq:ourm}
  (a,b\|z)+k\tsp\big(\slb+(a,z\|b)+(b,z\|a)\big)+\frac{k(k\m-1)}2\big(
  (a,c\|b)+(b,c\|a)\big) \ge 0.
\end{equation}
As discussed at the end of Section \ref{sec:all}, inequalities in
(\ref{eq:ourm}) can also be proved by induction on $k$, similarly to the
proof of Theorem \ref{thm:fminf}. This means that they can be obtained by a
different iterative method, exploiting the simple consequence of the Maximum
Entropy Principle used above:

\smallskip\hangindent=\parindent
Suppose that the inequality $\psi$ on $abcdz$ does not contain entropy terms
that intersect both $cd$ and $z$. Then $\psi$ is a valid entropy inequality
if and only if it holds for every entropic polymatroid that additionally
satisfies $(cd,z\|ab)=0$.

\smallskip
\noindent
Since entropic polymatroids surely satisfy the Shannon inequalities, the
first iteration extracts the additional inequalities that hold in
polymatroids with $(cd,z\|ab)=0$. This calculation has been done and
provided two types of inequalities:
$$
    \slb+(a,b\|z)+(a,z\|b)+(b,z\|a) \ge 0,
$$
and their symmetric versions; see \cite{Csirmaz.oneadhesive} or
\cite{M.fmadhe}. Consequently, they necessarily hold in entropic
polymatroids. In the second iteration, inequalities from the first iteration
are added to the Shannon inequalities, and then the consequences of this
larger set are extracted, and so on.

Even completing the second iteration seems to be beyond reach since
inequalities from the first iteration can be applied to any five subsets of
the base set, resulting in an enormous set of inequalities. Restricting the
substitution for permutations of the base elements, or certain one- and
two-element sets, reduces the number of inequalities to a manageable level.
For this and several other sparse sets, the computation has been done. The
obtained new inequalities contained, as expected, the corresponding
inequalities from (\ref{eq:m}) and (\ref{eq:ourm}); however, they also
contained many others that were clearly not tight entropy inequalities: they
were significantly weaker than those provided by Theorem \ref{thm:mainI}.
These preliminary results indicated no additional, structured family of
non-Shannon inequalities provided by this alternate method.

\section{Conclusions}\label{sec:final}

Structural properties of the entropy region of four or more variables are
mostly unknown. This region is bounded by linear inequalities corresponding
to the non-negativity of Shannon information measures. Finding additional
entropy inequalities is, and remains, an intriguing open problem. Previous
works on generating and applying such non-Shannon entropy inequalities
focused mainly on the four-variable case
\cite{Beimel.Orlov,csirmaz.book,DFZ11,studeny}, and only a few sporadic
five-variable non-Shannon inequalities have been discovered \cite{MMRV}.
This work provides infinitely many five-variable non-Shannon information
inequalities by systematically exploring a special property of entropic
vectors. Other works utilized the \emph{Copy Lemma}, a method distilled from
the original Zhang-Yeung construction by Dougherty et al.~\cite{DFZ11}. Our
method is based on a different paradigm derived from the principle of
maximum entropy and is a special case of the Maximum Entropy Method
described in\cite{exploring}. As proven in Lemma \ref{lemma:maxent1}, the
principle of maximum entropy implies that every entropic polymatroid has an
$n,m$-copy, which is a polymatroidal extension with special properties as
defined in Definition \ref{def:maxent1}. In Claim \ref{claim:mempoly}, we
have proved that polymatroids having $n,m$-copies form a polyhedral cone and
hint at how its facets can be computed. Facet equations provide the
potentially new non-Shannon entropy inequalities.

\smallskip

While the polyhedral computation presented in Claim \ref{claim:mempoly} is
numerically intractable even for small parameter values, the theoretical
results of Section \ref{sec:prepare} allowed us to reduce this complexity
significantly. Computational aspects of determining the facets of a
high-dimensional cone are closely related to linear multi-objective
optimization \cite{inner}. We have developed a specially tailored variant of
Benson's inner approximation algorithm \cite{inner,bensolve}, which takes
advantage of the special properties of this enumeration problem.
Computational results are reported in Section \ref{sec:compresults} for
generations $n\le 9$. Numerical instability, originating from both the
underlying LP solver and the polyhedral algorithm, prevented the completion
of the computation for larger values of $n$.

\smallskip

Non-Shannon inequalities obtained from these computations are discussed in
Section \ref{sec:ineq}. Based on these experimental results, two infinite
families of five-variable inequalities were defined. The first family in
Theorem \ref{thm:mainI} is parametrized by downward closed subsets of
non-negative lattice points. The second family in Theorem \ref{thm:mainII}
has a single positive integer parameter. Inequalities in both families are
\emph{proved} to hold for polymatroids on five elements that have an
$n$-copy; consequently, they are all valid entropy inequalities. It is
\emph{conjectured} that they cover all inequalities that can be obtained by
the applied method. In other words, if a polymatroid on five elements
satisfies all these inequalities, then it has an $n$-copy for all $n$. This
conjecture is left as an open problem. The computational results confirmed
this conjecture up to $n=9$.

Inequalities in the first family are investigated in Section
\ref{sec:reducedset} in more detail. They are specified by triplets
$\<\alpha_s,\beta_s,\gamma_s\>$ determined by downward closed sets $s$ of
nonnegative lattice points as discussed in Definition \ref{def:triplets}.
Such a triplet is \emph{extremal} if the corresponding inequality is not a
consequence of other inequalities from the same family. Extremal triplets
are determined by a special collection of downward closed sets called
\emph{irreducible staircases}. Based on the theoretical results in Corollary
\ref{corr1} and Claim \ref{claim:cases}, an incremental algorithm, sketched
as Algorithm \ref{code:1} was used to generate all irreducible staircases up
to generation $60$. The converse implication, which is valid for the
computed cases, that triplets generated by irreducible staircases are
extremal, is left as an open problem. Triplets
($\alpha_s,\beta_s,\gamma_s\>$ in the range $\beta_s,\gamma_s\le
20\tsp\alpha_s$, generated by irreducible staircases, are plotted in Figure
\ref{fig:ext1}. The number of new irreducible staircases that remained
irreducible in the subsequent generation matches the sequence A103116 in the
Encyclopedia of Integer Sequences \cite{oeis}. It is an interesting open
problem to prove the equality of these sequences.

To illustrate how the newly discovered entropy inequalities delimit the
five-variable entropy region, entropy vectors were normalized to satisfy
$(a,b\|z)=1$ and projected onto a three-dimensional subspace. Part of the
projection in the non-negative octant is denoted by $\Delta^+$. The Shannon
inequalities do not provide any restriction on this part. Figure
\ref{fig:region} illustrates the bounds implied by the new inequalities.
While the non-negative quadrant of the $xy$ plane is known to be part of
$\Delta^+$, and that it also contains points above that plane, it is an
intriguing open problem whether our bound is, at least asymptotically, tight
around the $x$ and $y$ axes. Showing that our bound is asymptotically tight
at the zero point would amount to settling the long-standing open problem of
whether the entropic region is semi-algebraic.

\section*{Funding}

The research reported in this paper was partially supported by the ERC
Advanced Grant ERMiD.


\end{document}